\documentclass{article}
\usepackage{amssymb} % that is for real numbers etc.
\usepackage{amsmath}
\usepackage{amsfonts}
\usepackage{natbib}
\usepackage{amsthm}   % enables to number all with 1 single counter
\usepackage{epsfig}
\usepackage{subfigure}
\RequirePackage[colorlinks,citecolor=blue,urlcolor=blue]{hyperref}

\newfont{\handw}{cmmi10 scaled 1200}

\newtheorem{Prop}{Proposition}[section]

\newtheorem{Th}[Prop]{Theorem}

\newtheorem{Rm}[Prop]{Remark}
\newtheorem{Def}[Prop]{Definition}

\newfont{\smcal}{cmu10 scaled 1200}

\newcommand{\cov}{\operatorname {COV}}

\newcommand{\E}{\operatorname {\mathbb E}}

\newcommand{\tr}{\operatorname {tr}}

\newcommand{\lw}{\mbox{\handw \symbol{96}}}
\newcommand{\argmin}{\operatorname {argmin}}

\title{ Intrinsic Inference on the Mean Geodesic of Planar Shapes and Tree Discrimination by Leaf Growth}
\author{Stephan Huckemann}
\date{}

\begin{document}
\maketitle{}

\begin{abstract} 
For planar landmark based shapes, taking into account the non-Euclidean geometry of the shape space, a statistical test for a common mean first geodesic principal component (GPC) is devised. It rests on one of two asymptotic scenarios, both of which are identical in a Euclidean geometry. 
For both scenarios, 
strong consistency and central limit theorems are established, along with an algorithm for the computation of a Ziezold mean geodesic. In application, this allows to verify the geodesic hypothesis for leaf growth of Canadian black poplars and to discriminate genetically different trees by observations of leaf shape growth over brief time intervals. 
With a test based on Procrustes tangent space coordinates, not involving the shape space's curvature, neither can  be achieved.
\end{abstract}
\par
\vspace{9pt}
\par
\vspace{9pt}
\noindent {\it AMS 2000 Subject Classification:} \begin{minipage}[t]{6cm}
Primary 62H30\\ Secondary 62H35, 53C22
 \end{minipage}

\section{Introduction}\label{intro-scn}

	In this paper the novel statistical problem of developing asymptotics for the estimation of the mean geodesic on a shape space is considered. It is the generalization to a non-Euclidean geometry of the asymptotics for the estimation of a straight  first principal component line from multivariate data in the Euclidean geometry. Due to curvature involved, however, methods from linear algebra as employed in the Euclidean geometry cannot be used, and a new approach has to be developed. The task at hand is more involved, yet somehow comparable to the situation of generalizing the concept of the mean for multivariate data to a mean for manifold valued data. For such manifold valued means pioneering work for definitions, existence, uniqueness, algorithms and asymptotics has been done by \cite{Gow, Z77,KWS90, G91,HL96,HL98,L01,BP03,BP05} and many others. In this work definitions for a mean geodesic, an algorithm and asymptotics are proposed and developed for data on Kendall's space of planar shapes. In particular, the following two different statistical scenarios are considered:  asymptotics with respect to underlying shapes -- the \emph{mean geodesic of shapes} -- and asymptotics with respect to underlying sampled geodesics -- the \emph{mean geodesic of geodesics}.

	The study of geodesics on shape spaces as the simplest model for a path of temporal evolution of shape is of high interest in shape analysis, in particular, in biological studies comparing growth patterns.

	Unlike previous attempts in the literature (e.g. \cite{JK87,KMMA01,KDL07} ) building on a Euclidean tangent space linearization of the shape space, the mean geodesic of geodesics defined here builds on a Euclidean tangent space linearization of the \emph{space of geodesics} of the shape space which has been introduced in \cite{HT06}. Hence as a new and quite abstract concept, we treat here geodesics as data points.

	In application, in a joint research study on leaf growth with the Institute for Forest Biometry and Informatics at the University of G\"ottingen, it turns out that it is precisely this subtle difference of linearizing the space of geodesics and not the shape space that successfully allows to discriminate genetically different Canadian black poplars by observation of leaf shape growth during a short time interval of the growing period. The research study presented here is %ground laying both for 
	fundamental for model building of leaf shape growth as well as for designing effective subsequent studies to investigate multiple endogenous and exogenous factors in leaf shape growth: E.g. since the beginning of the last century it has been well known that the leaf shape of (genetically) identical trees varies along a climate gradient (e.g. \cite{Brenner1902,BS15,RMRA08}). %, even for genetically identical trees (cf. \cite{RMRA08}). 
	Since \cite{Wolfe78} this  relationship has been successfully exploited for paleoclimate reconstruction resulting in the ``Climate Leaf Analysis Multivariate Program'' (CLAMP, \cite{Wolfe93}). Naturally, the underlying studies have been based on the shape of mature leaves; little is known about the temporal evolution of shape along a climate gradient. The research presented here indicates that a study involving only very few measurements of growing leaves %while they grow 
	may allow for a fairly good reconstruction and analysis of growth patterns, further elucidating the relationship of climate and leaf shape.

	This paper is organized in a theoretical and an applied part. % aiming at facilitating that each part may be read alone. 

	The theoretical first part consisting of the following two sections establishes the  statistical theory for the two types of means. In Section \ref{SS:scn}, after a brief review of Kendall's space of planar shapes, the concept of a Fr\'echet mean is extended to the space of geodesics while the underlying random deviates assume values in the shape space. Strong consistency in the sense of \citet{Z77} a well as in the sense of \citet{BP03} are established. In the appendix it shown that the original arguments can be extended nearly one-to-one to the general case considered here. In order to apply the central limit theorem (CLT) of 
	\cite{H_Procrustes_10}, smoothness in geodesics of the square of the canonical distance between shapes and geodesics for geodesics close to the data is established. Then in Section \ref{Asymp:scn}, smoothness is shown for the square of a metric of Ziezold type (cf. \cite{H_Procrustes_10}) for the space of geodesics leading to the other CLT. Finally, after establishing an explicit method for optimal positioning a fast algorithm for the computation of a mean geodesic of geodesics is derived. An algorithm for the mean geodesic of shapes has been derived earlier (\cite{HT06}). 

	The applied second part %consisting of three more sections 
	introduces the leaf shape data considered, the driving questions from forest biometry, statistical tests and some answers through the data analysis. In Section \ref{Forest:scn} the problem of discrimination by short growth observations is discussed. In particular, the relevance of the geodesic hypothesis from \cite{LK00} is noted for the devising of statistical tests in Section \ref{Tests:scn}. These are evaluated in Section \ref{App:scn} showing that only the \emph{test for common geodesics} can establish the validity of the geodesic hypothesis and the discrimination of genetically different trees on the basis of observations of brief leaf shape growth. Section \ref{Disc:scn} concludes with a discussion and gives an outlook.
 
\section{The  First Geodesic Principal Component for Planar Shape Spaces}\label{SS:scn}

	Throughout this work $\mathbb E(Y)$ denotes the classical expectation of a random variable $Y$ in a Euclidean space $\mathbb R^D$, $D\in \mathbb N$. A \emph{distance} $\delta$ on a topological space $\Gamma$ is a continuous mapping $\delta:\Gamma\times \Gamma\to[0,\infty)$ that vanishes on the diagonal $\{(\gamma,\gamma): \gamma\in \Gamma\}$; in contrast to a metric, $\delta$ is neither required to be non-zero off the diagonal, to be symmetric nor to satisfy the triangle inequality. 

	\paragraph{Kendall's planar shape spaces}
	In the statistical analysis of similarity shapes based on landmark configurations, geometrical $m$-dimensional objects (usually $m=2,3$) are studied by placing $k>m$ \emph{landmarks} at specific locations of each object, cf. Figure \ref{quadrangular-shape:fig} on page \pageref{quadrangular-shape:fig}. Each object is then described by a matrix in the space $M(m,k)$ of $m\times k$ matrices, each of the $k$ columns denoting an $m$-dimensional landmark vector. The usual inner product is denoted by $\langle x,y\rangle := \tr(xy^T)$ giving the norm $\|x\| = \sqrt{\langle x,x\rangle}$. For convenience and without loss of generality for the considerations below, only \emph{centered} configurations are considered. Centering can be achieved by multiplying with a sub-Helmert matrix 
	from the right, yielding a configuration 
	in $M(m,k-1)$. For this and other centering methods cf. \citet[Chapter 2]{DM98}. 
	Excluding also all configurations with all landmarks coinciding 
	gives the space of \emph{configurations} 
	\begin{eqnarray*}%\label{configuration_space_def}
	F_m^k&:=& M(m,k-1) \setminus \{0\} \,.
	\end{eqnarray*}
	Since 
	only the similarity shape is of concern, in particular we are not interested in size, we may assume that all configurations are contained in the \emph{pre-shape sphere}
	$ S_m^k :=\{x\in M(m,k-1): \|x\|=1\}$. Then, all normalized configurations that are related by a rotation from the special orthogonal group $SO(m)$ form the equivalence class of a \emph{shape}
	$$[x] = \{gx:g\in SO(m)\}$$
	and the canonical quotient is \emph{Kendall's shape space} 
	$$\Sigma_m^k := S_m^k/SO(m) = \{[x]:x\in S_m^k\},~~\mbox{with canonical projection }\frak{p}: S_m^k \to \Sigma_m^k\,.$$

	 %\mbox{ with the \emph{fiber} } [x] = \{gx:g\in SO(m)\}\,.$$
% 	Here, $SO(m) = \{g\in O(m): \det(g)=1\}$ denotes the special orthogonal group and $O(m) = \{g\in M(m,m): g^Tg=e\}$ the orthogonal group with the unit matrix $e=\diag(1,\ldots,1)$. The canonical projection is denoted by $\pi : S_m^k \to \Sigma_2^k$.

	In this paper we restrict ourselves to planar configurations, i.e. to the case of $m=2$. Then, complex notation comes in handy. For a detailed discussion of the following setup, cf. \cite{K84,K89} as well as \cite{KBCL99}. We take the notation from \cite{HT06}. 
	Identify $F_2^k$ with $\mathbb C^{k-1}\setminus \{0\}$ such that every landmark column corresponds to a complex number. This means in particular that $z\in \mathbb C^{k-1}$ is a complex row(!)-vector. With the Hermitian conjugate $a^* = (\overline{a_{kj}})$ of a complex matrix $a=(a_{jk})$ the pre-shape sphere $S_2^k$ is identified with $\{z\in \mathbb C^{k-1}: zz^*=1\}$ on which $SO(2)$ identified with $S^1=\{\lambda \in\mathbb C: |\lambda|=1\}$ acts by complex scalar multiplication. Then the well known Hopf-Fibration mapping to complex projective space gives $\Sigma_2^k=S_2^k/S^1=\mathbb CP^{k-2}$. %Note that 
%	$$ z w^* = \re(z w^*) + i \im(z w^*) = \langle z,w \rangle + i \langle z, iw\rangle =  \langle z,w \rangle-i\langle iz, w\rangle\,.$$

	\paragraph{The spaces of geodesics}  Note that every geodesic can be parametrized by unit speed, which we assume in the following. 
	Every great circle $\gamma(t) = x\cos t+ v\sin t$, $x,v\in S_2^k, \langle x,v\rangle=0$ is a geodesic on $S_2^k$, the space of geodesics is denoted by $\Gamma(S_2^k)$. A great circle is called a  \emph{horizontal great circle} if additionally $\langle ix,v\rangle =0$, the space of \emph{horizontal great circles} is denoted by $\Gamma^H(S_2^k)$. It is well known (e.g. \cite{KBCL99,HT06}) that this space projects to the space $\Gamma(\Sigma_2^k)$ of  geodesics of the shape space via
	$$ \Gamma(\Sigma_2^k) = \{\frak{p} \circ \gamma : \gamma \in \Gamma^H(S_2^k)\}\,.$$
	Then, with the \emph{Stiefel manifold} (giving all great circles)
	$$O_2(2,k-1)=\{(x,v)\in F_2^k\times F_2^k: \langle x,x\rangle = 1 = \langle v,v\rangle, \langle x,v\rangle = 0\}$$
%	and the function
% 	At every point $z\in  S_2^k \subset \mathbb  C^{k-1}$, the tangent space $T_zS_2^k$ decomposes into the \emph{vertical space} $T_z[z]$ along the fiber $[z]$ and its orthogonal complement $H_zS_2^k$, the \emph{horizontal space}: 
% 	$$T_zS_2^k = T_z[z]\oplus H_zS_2^k\,.$$
% 	$T_z[z]$ is spanned by $iz$, in fact, the fiber is the point set traversed by the vertical great circle $t\mapsto z\cos t + iz \sin t$. In consequence, every great circle orthogonal to one vertical great circle is orthogonal to all vertical great circles encountered. Such great circles are called \emph{horizontal}, their space is denoted by $\Gamma^H(S_2^k)$. With $\Gamma(\Sigma_2^k)$, the space of all geodesics on $\Sigma_2^k$ we have that 
% 	$$ \Gamma(\Sigma_2^k) = \{\pi \circ \gamma : \gamma \in \Gamma^H(S_2^k)\}\,.$$
% 	Recall that every great circle $t\mapsto x\cos t + v\sin t$ is determined by an offset $x$ and an orthogonal initial direction $v$. Then, with the function
% 	\begin{eqnarray}\label{Phi:eq} \Phi(x,v) &=& \left(\begin{array}{c}1-\langle x,x\rangle\\ 1-\langle v,v\rangle\\ 2\langle x,v\rangle \\ 2\langle x,iv\rangle\end{array}\right)
% 	\end{eqnarray}
	every tuple in the implicitly defined submanifold (additionally requiring horizontality)
	$$O^H_2(2,k-1) := \{(x,v)\in O_2(2,k-1): \langle x,iv\rangle %\Phi(x,v)
	=0\}$$
	%of codimension $1$ %of the orthonormal Stiefel manifold $O_2(2,k-1)=\{x,v\in M(2,k-1): \langle x,x\rangle = 1 = \langle v,v\rangle, \langle x,v\rangle = 0\}$, 
	corresponds to an element in $\Gamma^H(S_2^k)$.  Several tuples, however, may determine the same geodesic. To this end consider the action of the orthogonal group $O(2)$ and $S^1$ from right and left, respectively, by $(x,v) \mapsto  h_t\,(x,v) \,g_{\phi,\epsilon}$ with 
	\begin{eqnarray*}
	g_{\phi,\epsilon} = \left(\begin{array}{cc}\cos \phi&-\epsilon\sin\phi\\ \sin\phi&\epsilon\cos\phi\end{array}\right)\in O(2),&&h_t = e^{it}\in S^1,\end{eqnarray*}
	for $\phi,t \in [0,2\pi)$ and $\epsilon = \pm 1$ defined by 
	\begin{eqnarray*}
	(x,v)g_{\phi,\epsilon} &=& (x\cos\phi+v\sin\phi, v\epsilon\cos\phi-x\epsilon\sin\phi),\\
	h_t(x,v) &=&  (e^{it}x,e^{it}v)\,.\end{eqnarray*}
	For a manifold $M$ with a group $K$ acting from the right and a group $G$ acting from the left, denote by
	$$ K\backslash M/G = \{[P],P\in M\},\mbox{~~where~}[P] = \{gPk: g\in G, k\in K\}$$	
	the canonical double quotient, (e.g. \cite{T88}) %\citet[p. 149, Exercise 19 (a)]{T88}). %For historical rethe order is reversed
	With this notation we take the following from \cite{HT06}.

		\begin{Th}\label{sp_of_geod_thm} The space of point sets of all geodesics on planar shape space
		can be given
		the canonical structure
		$$\Gamma(\Sigma^{k}_2) ~\cong~ O(2)\backslash O^H_2(2,k-1)/S^1$$
		of a compact manifold of dimension $4k-10.$
		\end{Th}

	In analogy to the naming of the \emph{pre-shape sphere} $S_2^k$ call $O^H_2(2,k-1)$ the space of \emph{pre-geodesics}.

	A simpler argument yields $\Gamma(S_2^k)$ 
	as the compact manifold
	\begin{eqnarray}\label{space-spherical-geod:eq}
	\Gamma(S_2^k)&\cong&O(2)\backslash O_2(2,k-1)\,.\end{eqnarray}
	
	\paragraph{Distance from shapes to geodesics and between geodesics}
	The spherical distance $r(p,\gamma) = \arccos\sqrt{\langle p,x\rangle^2 + \langle p,v\rangle^2}$ of a point $p$ to the geodesic $\gamma$ defined by $(x,v)\in O_2^H(2,k-1)$  naturally defines a distance 
	$$\rho\big([p],\frak{p} \circ\gamma) = \min_{e^{it}\in S^1} r(e^{it}p,\gamma)$$
	 of the shape $[p]$ to the geodesic $\frak{p} \circ\gamma$ in the shape space. For a shape $[p]\in \Sigma_2^k$ denote by 
	$$\Gamma^{\pi/4}_{[p]} = \left\{\gamma \in \Gamma(\Sigma_2^k): \rho\left([p],\gamma \right)<\frac{\pi}{4}\right\}$$ 
	the open set of geodesics closer to $[p]$ than $\pi/4$. The proof of the following Theorem \ref{rhosq-smoot:th} is deferred to Appendix \ref{rhosq-smoot:ap}.

	\begin{Th}\label{rhosq-smoot:th}
	 For fixed $p\in  S_2^k$ the function
	$$  \gamma \mapsto \rho\big([p],\gamma)^2$$
	is smooth on $\Gamma^{\pi/4}_{[p]}$.
	\end{Th}

	In order to measure the distance between geodesics equip $O(2)\backslash O_2^H(2,k-1)$ with a suitable Riemannian  structure -- two of such structures are straightforward, cf. \cite{EAS98}, or more simply, embed $O_2^H(2,k-1)$ in a Euclidean space and consider the quotient distance w.r.t. to the corresponding extrinsic metric. More precisely, we generalize a setup introduced by \cite{Z94} on the quotient $\Sigma_2^k$.

	\begin{Def}\label{Ziez_dist:def}
	The Euclidean distance
	$$ d(P,Q) := \sqrt{\|x-y\|^2+ \|v-w\|^2}\,$$
	for $P=(x,v), Q=(y,w) \in O_2^H(2,k-1) \subset  F_2^{k} \times  F_2^{k}$ defines the canonical quotient distance
	$$ \delta([P],[Q]) := \min_{\footnotesize\begin{array}{c}h,h'\in O(2),\\g,g'\in S^1\end{array}}d(gPh,g'Qh')\,. $$
	for $[P] [Q] \in \Gamma(\Sigma_2^k)$ called the \emph{Ziezold distance} on $\Gamma(\Sigma_2^k)$.
	\end{Def}

	\paragraph{The mean geodesic of shapes} In earlier work (\cite{HHM07}) establishing a general framework for geodesic principal component analysis, the mean geodesic of shapes has been called a \emph{first geodesic principal component}.
	 
	\begin{Def}%[The mean geodesic of shapes]
	\label{rho-mean:def}
	Suppose that $X,X_1,\ldots,X_n$ are i.i.d. random pre-shapes mapping from an abstract probability space $(\Omega,{\cal A},{\cal P})$ to $S_2^k$ equipped with its Borel $\sigma$-algebra. For $\omega\in \Omega$ call the geodesics $\gamma_n(\omega),\gamma^* \in \Gamma(\Sigma_2^k)$ the \emph{first sample and population  geodesic principal component} (GPC) of the sample $[X_1(\omega)],\ldots, [X_n(\omega)]$ and $[X]$, respectively, if  
	\begin{eqnarray*} \sum_{j=1}^n\rho([X_j(\omega)],\gamma_n(\omega))^2 &=& \min_{\gamma\in \Gamma(\Sigma_2^k)}\sum_{j=1}^n\rho([X_j(\omega)],\gamma)^2\,, \mbox{ for all }\omega\in \Omega\,,\\
	\mathbb E\left(\rho([X],\gamma^*)^2\right) &=&  \min_{\gamma\in \Gamma(\Sigma_2^k)}\mathbb E\left(\rho([X],\gamma)^2\right)\,.
	\end{eqnarray*} 
	The random set of all sample GPCs is denoted by $E^{(\rho)}_n(\omega)$, $E^{(\rho)}([X])$ is the set of all population GPCs.	
	\end{Def}

	\begin{Th}[Asymptotics for the mean geodesic of shapes]\label{complex_proj_geode_SLLN:thm}
	For i.i.d. random pre-shapes $X,X_1,\ldots,X_n$ the set of first sample GPCs $E^{(\rho)}_n(\omega)$ is a uniformly strongly consistent estimator of the set of first population GPCs $E^{(\rho)}([X])$ in the sense that for every $\epsilon >0$ and a.s. for every $\omega\in \Omega$ there is a number $n(\epsilon,\omega) \in \mathbb N$ such that
	$$\bigcup_{j=n}^\infty E^{(\rho)}_j(\omega)\subset \Big\{\gamma \in \Gamma(\Sigma_2^k): \delta\big(\gamma,E^{(\rho)}([X])\big) \leq \epsilon\Big\}\,.$$
	Moreover, 
	if $E^{(\rho)}([X])$ contains a unique element $\gamma^*$ contained in
	$$ \bigcap_{p \in {\rm Supp}(X)}\Gamma^{\pi/4}_{[p]}$$
	with the support ${\rm Supp}(X)$ of $X$, if $\gamma_n \in E^{(\rho)}_n(\omega)$ is a measurable selection and  $x=\phi(\gamma)\in \mathbb R^{4k-10}$ are local coordinates near $\gamma^*$  with $\phi(\gamma^*)=0$, then 
	$$A\sqrt{n}~\phi(\gamma_n)~ \stackrel{}{\to}~ {\cal N}(0,\Sigma)\mbox{~~in distribution}$$
	with the $(4k-10)$-dimensional normal distribution ${\cal N}(0,\Sigma)$ with zero mean and covariance matrix
	$\Sigma = \cov\big({\rm grad}_x\rho(X,\gamma^*)^2\big)$ and $A = \E(H_x\rho(X,\gamma^*)^2\big)\big)$. Here, grad$_x$ and $H_x$ denote the gradient and Hessian of $\rho^2$, respectively, w.r.t. the coordinate $x$. 
	\end{Th}
	
	\begin{proof}
	The assertion of strong consistency is a consequence of the general Theorem \ref{ext_Ziezold:thm} and Theorem \ref{ext_BP:thm} in the appendix. Since $\rho([X],\gamma)^2$ is smooth in $\gamma$ as long as $\gamma$ is closer to $[X]$ than $\pi/4$, by Theorem \ref{rhosq-smoot:th}, the assertion of the central limit theorem (CLT) follows from the CLT of \citet[Theorem A.1]{H_Procrustes_10} since $\Gamma(\Sigma_2^k)$ is compact.
	\end{proof}

	\begin{Rm}\label{estimate_cov:rm}
	For practical applications of Theorem \ref{complex_proj_geode_SLLN:thm}, in a given chart %around a sample mean, 
	$A$ and $\Sigma$ could be estimated %from the data 
	by classical numerical and multivariate methods. Alternatively, estimates can be obtained simply from the data's covariance in a chart around a sample mean. In particular in case of non-singular $A$, that covariance tends asymptotically to $A^{-1}\Sigma (A^{-1})^T$.
	\end{Rm}

	\paragraph{Uniqueness and location of the first GPC}  The hypothesis of a unique first GPC is essential for the following framework. Clearly, that hypothesis translates to an anisotropy condition on the random shape. E.g. on the basis of the geodesic hypothesis for biological growth as detailed in Section \ref{Forest:scn}, we may assume uniqueness in the application in Section \ref{App:scn}. The development of a test for specific anisotropy would certainly be of merit for other potential applications. By definition, every first GPC will be close to the support of $[X]$. %, if the distribution of $X$ is sufficiently concentrated, then one can expect it will be contained  as well in $ \cap_{p \in {\rm Supp}(X)}\Gamma^{\pi/2}_{[p]}$. 
	Data analysis and numerical simulations show that the intrinsic mean is usually very close to the first GPC (e.g. \cite{HT06,HHM07}). Moreover, for sufficient concentration, the intrinsic mean is unique and contained in a ball around the support of radius $\pi/4$ (cf. \cite{KWS90,L01}). % which , guaranteeing that the above CLT is applicable. 
	Certainly, further research is necessary to tackle questions of uniqueness and location.

\section{The Ziezold Mean of a Random First GPC}\label{Asymp:scn}

	We are now in the situation of having samples of first sample GPCs and to determine their Fr\'echet mean w.r.t. to some distance. In order to apply a CLT we are aiming for a mean in a smooth sense. It turns out that the comparatively simple Ziezold distance features the desired smoothness. 

	\begin{Th}\label{Ziezold.Dist:th} The following hold:
	\begin{enumerate}
	 \item[(i)] the action of $S^1$ and $O(2)$ is isometric with respect to $d$, i.e. $d(P,Q)=d(gPh,gQh)$ for all $P,Q\in O_2^H(2,k-1)$ and all $g\in S^1,h\in O(2)$,
	\item[(ii)] $\delta^2$ is smooth and $\delta$ is a metric on  $\Gamma(\Sigma_2^k)$. 
	\end{enumerate}
	\end{Th}

	\begin{proof} Property (i) is easily verified, in fact, the left-action of $S^1$ and right-action of $O(2)$ are even isometric on the ambient $\mathbb C^{k-1} \times \mathbb C^{k-1}$. Moreover on $F_2^{k} \times F_2^{k}$ the isotropy groups are $\{(g_{0,1},h_0),(g_{\pi,1},h_{\pi})\}$. For this reason in consequence of the Principal Orbit Theorem (e.g. \citet[Chapter IV.3]{Bre72}), $\delta$ extends to the natural geodesic quotient metric on the manifold $O(2)\backslash(F_2^{k}\times F_2^{k})/S^1$. Hence in particular, $\delta^2$ is smooth on the submanifold $\Gamma(\Sigma_2^k)$. As another consequence, since the extension of $\delta$ is a metric, $\delta$ itself is a metric which yields (ii).
	\end{proof}

%	The proof of Theorem \ref{Ziezold.Dist:th} can be found in Appendix \ref{Asymp:app}. 

	In view of the application in Section \ref{App:scn}, we now consider samples of independent random geodesics obtained from not necessarily independent shapes as typically occur during observation of growth. In particular, the test for common geodesics devised in Section \ref{Tests:scn} relies on the following Theorem \ref{Zie_mean:thm}.

	\begin{Def}[The mean geodesic of geodesics]
	Call $\gamma^*\in \Gamma(\Sigma_2^k)$
	\begin{itemize}
	 \item[]
	a population \emph{Ziezold mean geodesic} of a random geodesic $\Xi$ if
	$$ \E\big(\delta(\Xi,\gamma^*)^2\big) = \min_{\gamma\in \Gamma(\Sigma_2^k)} \E\big(\delta(\Xi,\gamma)^2\big)\,,$$  
	\item[] a sample \emph{Ziezold mean geodesic} of random geodesics $\Xi_1,\ldots,\Xi_n$ if
	$$ \sum_{j=1}^n\delta(\Xi_j(\omega),\gamma^*)^2\big) = \min_{\gamma\in \Gamma(\Sigma_2^k)} \sum_{j=1}^n\delta(\Xi_j(\omega),\gamma)^2\big)\,.$$ 
	\end{itemize}
	The sets of population and sample Ziezold mean geodesics are denoted by $E^{(\delta)}(\Xi)$ and $E^{(\delta)}_n(\omega)$, respectively.
	\end{Def}

	\begin{Th}[Asymptotics for the mean geodesic of geodesics]\label{Zie_mean:thm}
	For i.i.d. random geodesics $\Xi,\Xi_1,\ldots,\Xi_n$ the set of sample Ziezold mean geodesics $E^{(\delta)}_n(\omega)$ is a uniformly strongly consistent estimator of the set of population Ziezold mean geodesics $E^{(\delta)}(\Xi)$ in the sense that for every $\epsilon >0$ and a.s. for every $\omega\in \Omega$ there is a number $n(\epsilon,\omega) \in \mathbb N$ such that
	$$\bigcup_{j=n}^\infty E^{(\delta)}_j(\omega)\subset \Big\{\gamma \in \Gamma(\Sigma_2^k): \delta\big(\gamma,E^{(\delta)}(\Xi)\big) \leq \epsilon\Big\}\,.$$
	If $E^{(\delta)}(\Xi)$ contains a unique element $\gamma^*$, $\gamma_n \in E^{(\delta)}_n(\omega)$ is a measurable selection and  $x=\phi(\gamma)\in \mathbb R^{4k-10}$ are local coordinates near $\gamma^*$  with $\phi(\gamma^*)=0$, then 
	$$A\sqrt{n}~\phi(\gamma_n)~ \stackrel{}{\to}~ {\cal N}(0,\Sigma)\mbox{~~in distribution}$$
	with the $(4k-10)$-dimensional normal distribution ${\cal N}(0,\Sigma)$ with zero mean and covariance matrix
	$\Sigma = \cov\big({\rm grad}_x\delta(\Xi,\gamma^*)^2\big)$ and $A = \E(H_x\delta(\Xi,\gamma^*)^2\big)\big)$. Here, grad$_x$ and $H_x$ denote the gradient and Hessian of $\delta^2$, respectively, w.r.t. the coordinate $x$. 
	\end{Th}
	
	\begin{proof}
	Since $\delta$ is a metric by Theorem \ref{Ziezold.Dist:th}, the assertion on strong consistency is a consequence of \citet{Z77} as \citet[Remark 2.5]{BP03} teach. Since $\delta$ is neither an intrinsic nor an extrinsic metric, the CLT of \cite{BP05} cannot be applied. Rather, the assertion of the CLT follows from the more general CLT of \citet[Theorem A.1]{H_Procrustes_10}, since by Theorem \ref{Ziezold.Dist:th}, $\delta^2$ is smooth and $\Gamma(\Sigma_2^k)$ is compact.
	\end{proof}

	For practical applications of Theorem \ref{Zie_mean:thm} proceed as detailed in Remark \ref{estimate_cov:rm}.

%\section{An Algorithm for the Ziezold Mean Geodesic}\label{Algo:scn}
	For $P,Q \in O^H_2(2,k-1)$, $g\in S^1$ and $h\in O(2)$ call $gQh$ is in \emph{optimal position} to $P$, if $d(P,gQh) = \delta([P],[Q])$. Since both groups $O(2)$ and $S^1$ are compact, given $P\in O_2(2,k-1)$, every $Q\in O_2(2,k-1)$, can be placed into optimal position to $P$. Moreover, if $[P^*]$ is the unique Ziezold mean geodesic of sampled geodesics $[P_1],\ldots,[P_n]$ then $P^*$ is the extrinsic mean of the $g_jP_jh_j$ placed into optimal position to $P^*$, $g_j \in S^1$, $h_j\in O(2)$, $j=1,\ldots,n$, i.e.
	$$P^* = \argmin_{P\in O_2(2,k-1)} \sum_{j=1}^n\min_{\footnotesize\begin{array}{c}h_j\in O(2),\\g_j\in S^1\end{array}}d(P,g_jP_jh_j)\,,$$
	cf. \cite{H_meansmeans_10}. The extrinsic mean then is the orthogonal projection to $O^H_2(2,k-1)$ of the classical Euclidean mean in ambient $M(2,k-1)\times M(2,k-1)$, cf. \cite{HL98}; \cite{BP03}.

	In the first step we solve the problem of optimally positioning analytically, in the second step we compute the orthogonal projection. Based on the two, the algorithm of \cite{Z94} is adapted, to compute the Ziezold mean geodesic. 

	\begin{Th}\label{opt_pos:th}
	 Let $P=(x,v),Q=(y.w) \in O_2^H(2,k-1)$ and define
	$$\begin{array}{rclcrcl}
	   A &:=& \langle x,y\rangle + \epsilon \langle v,w\rangle&,&B&:=&\langle x,w\rangle -\epsilon \langle v,y\rangle\,, \\
	C&:=&\langle x,iy\rangle + \epsilon \langle v,iw\rangle&,&D&:=&\langle x,iw\rangle - \epsilon\langle v,iy\rangle\,.
	  \end{array}$$
	Then, for $g_{\phi,\epsilon},h_t$ putting $Q$ into optimal position  $g_tQh_{\phi,\epsilon}$ to $P$, it is necessary that 
	$$\tan \phi = \frac{B  + D\tan t}{A  +  C \tan t} $$ and that $t$ satisfies
	\begin{enumerate}
	 \item[(i)] $$\tan t = \alpha \pm \sqrt{\alpha^2+1}, \mbox{ with } \alpha = \frac{C^2+D^2-A^2-B^2}{2(AC+BD)}$$ in case of $AC+BD\neq 0$,
	 \item[(ii)] $t=0$ in case of $AC+BD=0\neq A^2+B^2-C^2+D^2$.
	\end{enumerate}
	$-\pi/2\leq t<\pi/2 $ may be arbitrary in case of $AC+BD=0= A^2+B^2-C^2+D^2$.
	\end{Th}

	\begin{proof}
	\begin{eqnarray*}\lefteqn{
	d(P,g_tQh_{\phi,\epsilon})^2}\\% &=& \|x-e^{it}y\cos\phi  - e^{it}w\sin\phi \|^2 +  \|v-e^{it}w\epsilon\cos\phi  + e^{it}y\epsilon\sin\phi \|^2\\
	&=& 4 - 2 \Bigg(\cos\phi \Big(\langle x, e^{it}y\rangle +\epsilon\langle v, e^{it}w\rangle\Big) + \sin\phi \Big(\langle x, e^{it}w\rangle -\epsilon\langle v, e^{it}y\rangle\Big)\Bigg)
	\end{eqnarray*}	
	gives
	\begin{eqnarray*}\lefteqn{
	\frac{4 - d(P,g_tQh_{\phi,\epsilon})^2}{2}}\\
	&=& A \cos \phi \cos t + B \sin\phi \cos t + C \cos \phi \sin t+ D\sin \phi \sin t\,.
	\end{eqnarray*}	
	For fixed $\phi$, a necessary condition for $t=t(\phi)$ to maximize the above r.h.s. is that
	$$ \tan t = \frac{C \cos \phi + D\sin \phi}{A \cos \phi +  B \sin\phi} =  \frac{C  + D\tan \phi}{A  +  B \tan\phi}\,.$$
	Similarly, a necessary condition for $\phi = \phi(t)$ is that
	$$ \tan\phi  = \frac{B \cos t + D\sin t}{A \cos t +  C \sin t} =  \frac{B  + D\tan t}{A  +  C \tan t}\,.$$
	Letting $\zeta= \tan t, \eta=\tan\phi$ we obtain
	\begin{eqnarray*}
	 \zeta &=&  \frac{C  + D \eta}{A  +  B \eta}~=~\frac{C(A+C\zeta)  + D(B+D\zeta)}{A(A+C\zeta)  +  B(B+D\zeta)}
	\end{eqnarray*}
	and, equivalently 
	\begin{eqnarray*}
	 (AC+BD)\zeta^2 + (A^2+B^2)\zeta &=& (C^2+D^2)\zeta + (AC+BD)\,,
	\end{eqnarray*}
	yielding the assertion.
	\end{proof}

%	The proof of Theorem \ref{opt_pos:th} can be found in Appendix \ref{Asymp:app}.
	\begin{Th}\label{orth_proj:th}
	 Suppose that $P=(x,v) \in F_2^k\times F_2^k$, then $(\zeta,\eta) \in O_2^H(2,k-1)$ is the orthogonal projection of $P$ to $O_2^H(2,k-1)$ if and only if
	\begin{eqnarray*}
	 \zeta&=&\frac{1}{\langle x,\zeta\rangle}\Big(x - \langle x,\eta\rangle \eta - \langle x,i\eta\rangle i\eta\Big)\\
	 \eta&=&\frac{1}{\langle v,\eta\rangle}\Big(v - \langle v,\zeta\rangle \zeta - \langle v,i\zeta\rangle i\zeta\Big)
	\end{eqnarray*}
	$\zeta$ is arbitrary in case of $\langle x,\zeta\rangle=0$, and $\eta$ is arbitrary in case of $\langle v,\eta\rangle=0$. 
	\end{Th}

	\begin{proof}
	 Apply Lagrange minimization to $\|x-\zeta\|^2 + \|v-\eta\|^2$ for $\zeta,\eta \in F_2^k$ under the constraining condition $\Phi(\zeta,\eta)=0$ for 
	\begin{eqnarray*}%\label{Phi:eq} 
	\Phi(x,v) &=& \left(\begin{array}{c}1-\langle x,x\rangle\\ 1-\langle v,v\rangle\\ 2\langle x,v\rangle \\ 2\langle x,iv\rangle\end{array}\right)\,.
	\end{eqnarray*}
	%from equation (\ref{Phi:eq}).
	\end{proof}

	\paragraph{Algorithm to obtain a pre-geodesic of a Ziezold mean geodesic} Let $P_1,\ldots,P_J$ be a sample of pre-geodesics. Starting with an initial value $(x^{(0)},v^{(0)}) = P^{(0)} := P_1$, say, obtain $P^{(n+1)}=(x^{(n+1)},v^{(n+1)})$ from $P^{(n)}=(x^{(n)},v^{(n)})$ for $n=0,1,\ldots $ by putting all $P_j \,(j=1,\ldots,J)$ in optimal position $P^*_j=(y^*_j,w^*_j)$ to $P^{(n)}$ by computing the corresponding $\phi_j,t_j,\epsilon_j$ from Theorem \ref{opt_pos:th}. Then, set 
	$$\left(x,v\right) :=\frac{1}{J}\,\left(\sum_{j=1}^Jy^*_j, \sum_{j=1}^Jw^*_j\right)$$
	and let $P^{(n+1)}$ be the orthogonal projection of $(x,v)$ to $O_2^H(2,k-1)$ from Theorem \ref{orth_proj:th}.

\section{Leaf Growth  Data and Problem Statement}\label{Forest:scn} 

	\begin{figure}[h!]
	%\centering
	 \includegraphics[angle=-90,width=0.45\textwidth]{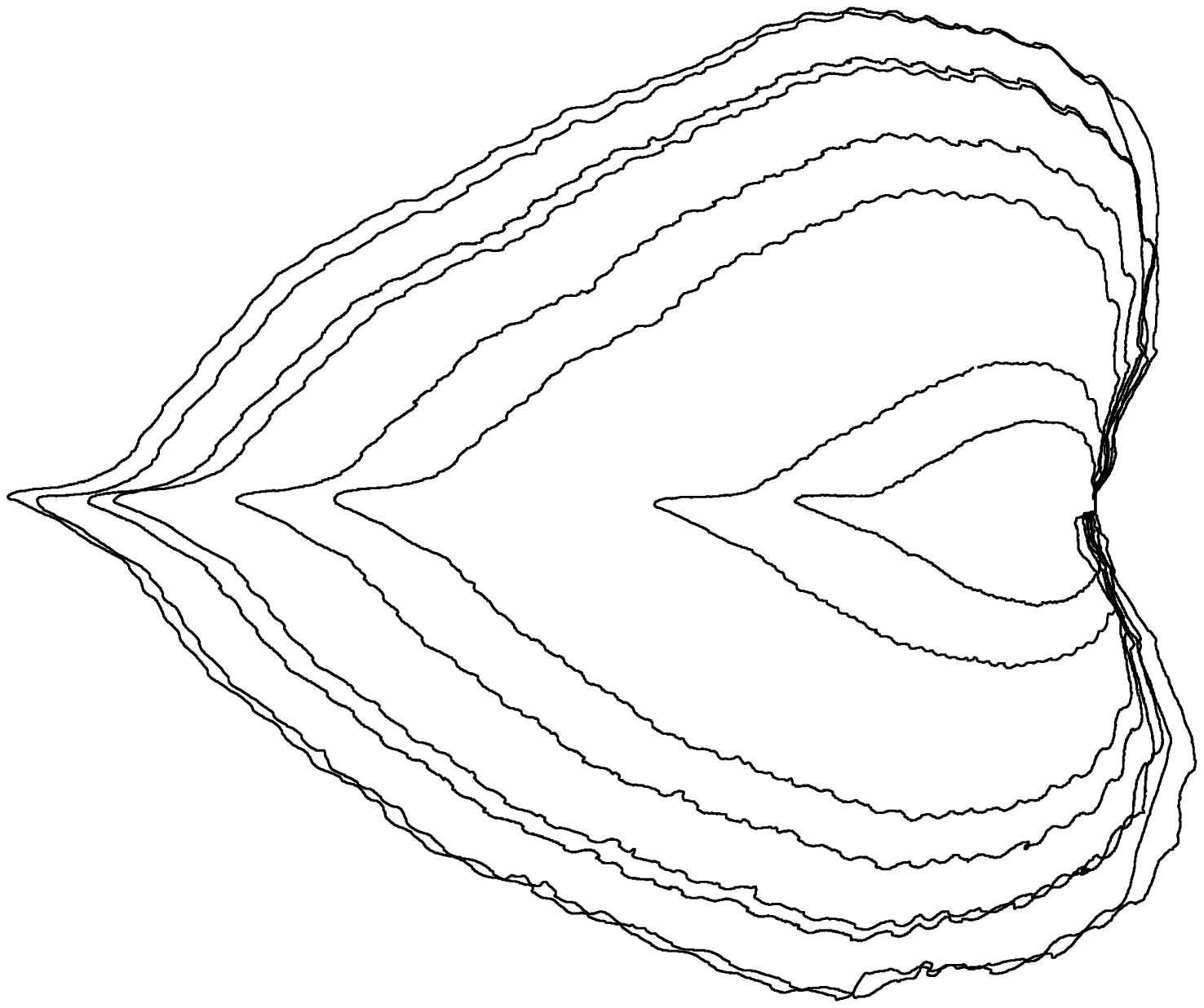}
	 \includegraphics[angle=-90,width=0.45\textwidth]{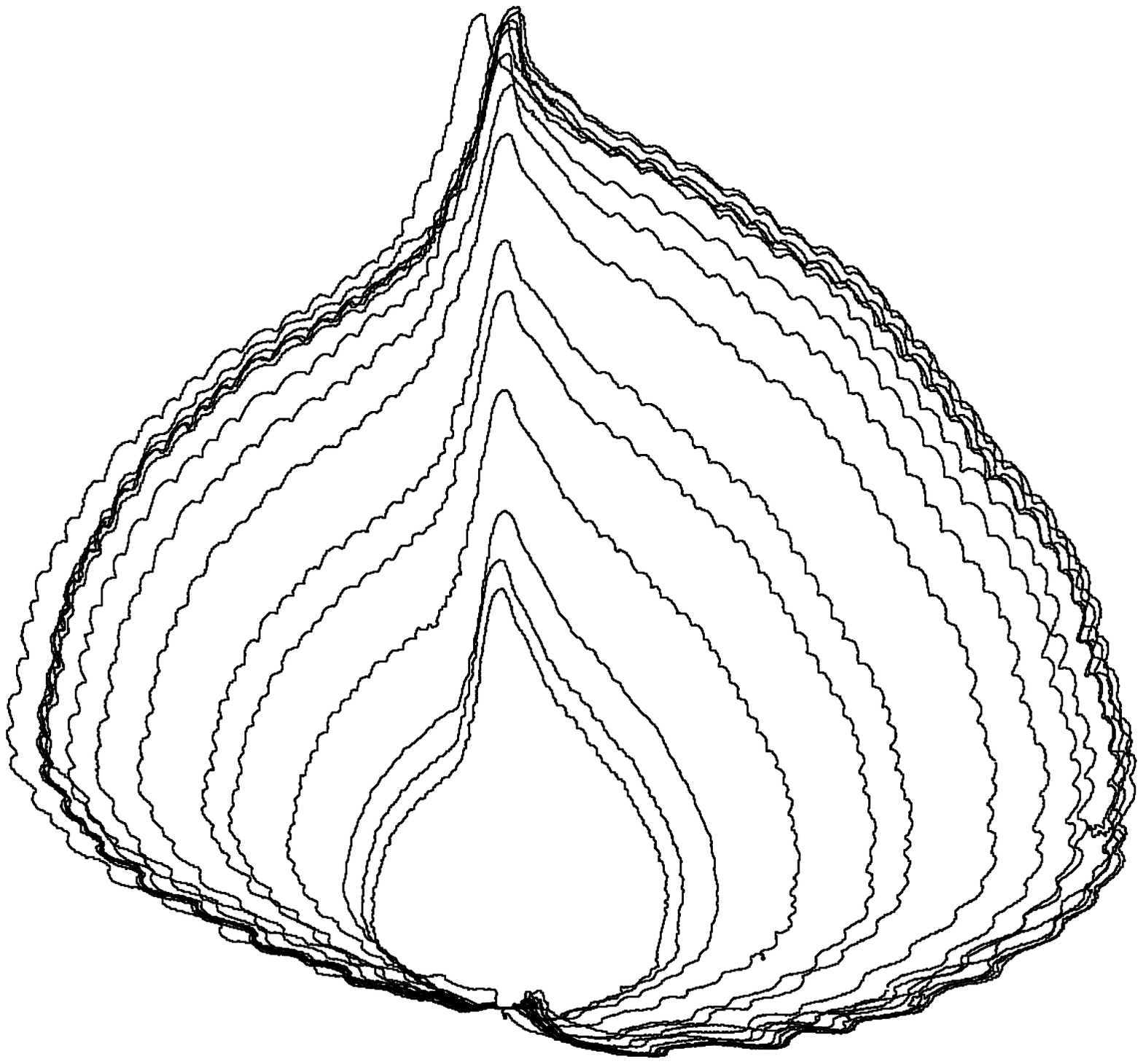}
	\begin{minipage}{0.45\textwidth}
	 \includegraphics[angle=-90,width=1\textwidth]{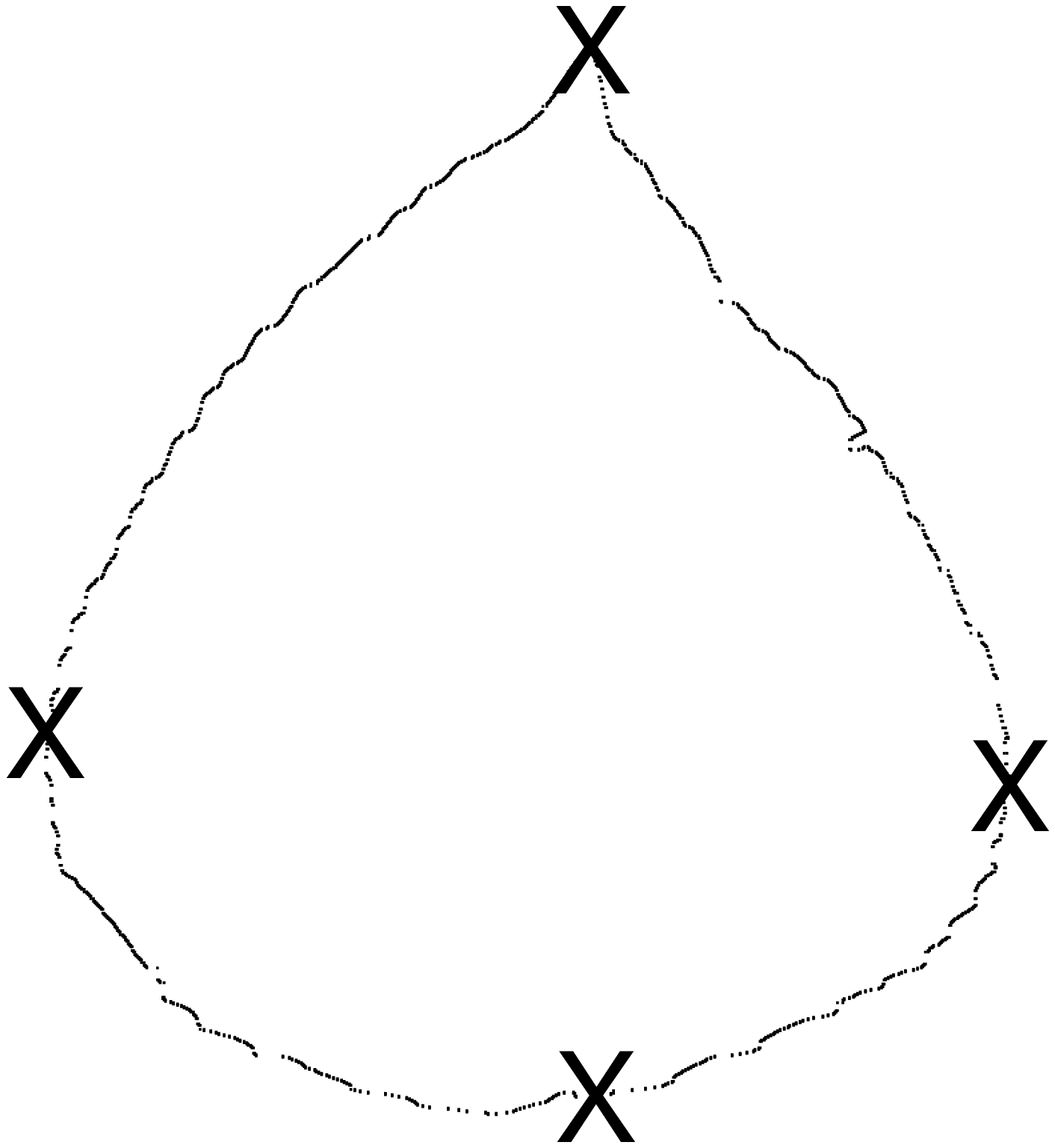}
	\end{minipage}
	\begin{minipage}{0.45\textwidth}
	\caption{\it Top row: typical leaf growth over a growing period of a reference tree (left) and one of the two clones (right). Bottom left: typical digitized leaf contour and landmarks of the corresponding quadrangular configuration at petiole, tip, and largest extensions orthogonal to the connecting line.\label{quadrangular-shape:fig}}
	\end{minipage}
	\end{figure}

	\paragraph{From leaf data to shape descriptors}  We consider leaf shape data collected from two clones and a reference tree of black Canadian poplars at an experimental site at the University of G\"ottingen. These data are similar but different from the data reported on in \cite{HHM09} and \cite{H_dynamics_10}. They  consist of the shapes of 21 leaves from clone 1 and of 11 leaves from clone 2 as well as of the shapes of 12 leaves from the reference tree, all of which have been recorded non-destructively over several days during a major portion of their growing period of approximately one month (the maximal number of observations is 17, the minimal 2 with a median of 13). The top row in Figure \ref{quadrangular-shape:fig} shows typical contours of  leaf growth. At each time point, from each leaf contour a quadrangular \emph{landmark based configuration} has been extracted by placing one landmark at the petiole (where the stalk enters the leaf blade), one at the leaf tip (the endpoint of the main leaf vein) and two, each at the maximal extensions orthogonal to the line connecting petiole and tip, cf. the bottom image of Figure \ref{quadrangular-shape:fig}. These four landmarks encode in particular the information of length, width and vertical and horizontal assymetry.  As detailed in Section \ref{SS:scn}, these landmarks additionally convey a correspondence of leaves with \emph{shapes}, i.e. points in the \emph{shape space} $\Sigma_2^4$. This space is a non-Euclidean manifold and a special case of \emph{Kendall's landmark based shape spaces}.

	\paragraph{The geodesic and parallel hypotheses for biological growth}
	Investigating landmark based configurations of rat skulls, \cite{LK00} observed that:
	\begin{center}{\it 
	the shape change due to biological growth mainly\\ follows a geodesic in Kendall's shape space.}
	\end{center}
	  In a research modeling the growth of tree-stem disks as well as leaf growth this \emph{geodesic hypothesis} has been corroborated by \cite{HHGMS07}.
	\cite{JK87, KMMA01, EDL, KDL07} have proposed more subtle models for shape growth essentially building on polynomials in Procrustes residuals (cf. Section \ref{Tests:scn} below). 

	%Using Procrustes tangent space coordinates in landmark based shape spaces 
	Additionally, \cite{MKMA00} observed parallel growth patterns and coined the \emph{parallel hypothesis}, stating that Procrustes residuals of related biological objects follow curves parallel in the Euclidean geometry of the tangent space at a Procrustes mean. In view of the geodesic hypothesis we restrict those curves %here 
	to straight lines, generally however, not mapping to geodesics (cf. Figure \ref{two-arrows:fig}). %, and % in the shape space, even though of course, geodesics on 
%	a compact space geodesics can never be parallel. 
	%Then, 
%	For sufficiently concentrated data, however, in approximation, the geodesic hypothesis implies the parallel hypothesis.
% 	, in the geometry of which we consider parallelity. Although parallelity can only be given a precise meaning for tangent space vectors, in approximation we can also speak of parallel geodesics segments. Here is the formulation of the \emph{parallel hypothesis} suited to this research:
% 	\begin{center}{\it 
% 		biological growth of related objects, possibly of initially and terminally \\different shape, tends to follow nearly parallel geodesic segments.}
% 	\end{center}
% %	A workable concept of parallelity of geodesic segments is made precise below in Defi\-nition \ref{parallel:def}. 
% 	Note that the geodesic hypothesis implies the parallel hypothesis. 

	\begin{figure}[h!]
	%\centering
	\begin{minipage}{0.57\textwidth}
	 \includegraphics[width=1\textwidth]{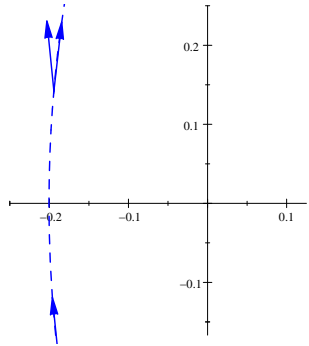} %two-arrows.eps}
	\end{minipage}
	\begin{minipage}{0.01\textwidth}\hfill \end{minipage}
	\begin{minipage}{0.4\textwidth}
	\caption{\it 
	Tangent space  of the two-dimensional $\Sigma_2^3$ (obtained from $\Sigma_2^4$ by leaving out one landmark) under the inverse Riemann exponential at the intrinsic mean corresponding to the extent of the overall data (clones 1 + 2 and reference tree). % view in Figure \ref{pc_shapes:fig}. 
	The left top vector is the affine parallel transport of the bottom vector in the Euclidean tangent space, the right top vector is its intrinsic parallel transplant along a geodesic (dashed). 
% 	n intrinsic of a quarter Tangent vector parallel along a geodesic 
% 	Tangent vectors parallel along a unit circle at points of distance 0.15 radians isometrically approximating the shape space. This is approximately half the maximal distance between shape of the two groups in Figure \ref{pc_shapes:fig}.
	\label{two-arrows:fig}}
	\end{minipage}
	\end{figure}

	\paragraph{A brief discussion of the geodesic hypothesis} %and parallel hypotheses} 
	In D'Arcy Thompson's seminal work \cite{T17}, biological form and growth of form has been explained by the invocation of the mathematical concpet of force. More recently, the relationship between growth and energy minimization has been explored by \cite{B78}. These works have led \cite{LK00} %as well as \cite{MKMA00} 
	to the above hypothesis, %es 
	%certainly % Of course, the  %latter 
	%authors have 
	being aware that, firstly, % of the following. %two 
	%basic 
	%facts. 
	%Firstly, 
	geodesics depend on a specific geometry %and that there seems to be no unique canonical geometry 
	of a shape space, and secondly, even though many paths of growth seem to follow geodesics, there are examples where the geodesic fit is rather poor (e.g. \cite{EDL}). E.g. for the space of planar triangles, the hyperbolic geometry of the complex upper half plane introduced by \cite{B86} seems just as natural as the spherical geometry of the complex projective space in one complex dimension (Kendall's shape space for planar triangles, cf. \cite{K84}). % also, for certain biological settings, environmental factors may be dominating and too diverse. 
	However, considering geodesics in Kendall's planar shape space as a rough working hypothesis, in particular for the leaf shapes in question, seems like a promising starting point for statistical investigation in the same way that approximate linearity has served statisticians well since the time of Gauss (or even earlier).

	\paragraph{Problem statement} As visible in Figure \ref{quadrangular-shape:fig}, the shape of leaves of the clones can usually be well discriminated from the shape of leaves of the reference tree by visual inspection. 	Following the geodesic hypothesis, the shape change under growth could be predicted from initial observations, ideally two initial observations would suffice. Since for the data at hand, the evolution of leaf contours have been followed elaborately along several time points,  then the effort for future research could be cut down considerably. This leads to the following fundamental problem. 
% 	\begin{description}
% 	 \item[Problem:] 
	\begin{center}{\it 
	Can leaf growth of genetically identical trees be predicted\\ and discriminated from growth of genetically different trees\\ on the basis of few initial measurements? }
	\end{center}
%	 \item[Problem 2:] What is the extend of parallelity in leaf growth of genetically different specimen? 
%	\end{description}
	In this research we restrict ourselves to measuring shape by four landmarks as detailed above.

\section{Tests for Shape Dynamics}\label{Tests:scn}
		
	The precise definitions used in this section can be found in Sections \ref{SS:scn} and \ref{Asymp:scn}. 
	For every leaf considered a \emph{first geodesic principal component} (GPC -- a generalization  to manifolds of a first principal component direction) is computed either from its first two shapes (then the GPC is just the geodesic connecting the two, the correspondig group is called ``young'') or from the rest of the shapes (that goup is called ``old'') over the growing period (for an algorithm, see \cite{HT06}). For two groups (young vs. young, young vs. old and old vs. old of different trees) to be tested for a common mean first GPC, the procedure detailed below produces data in a Euclidean space, namely in the tangent space of the space of geodesics at a mean geodesic. For comparison, a classical test for equality of mean shape as well as a test for equality of mean direction based on classical methodology described below, similarly produce data in a Euclidean space, namely the tangent space of the shape space at a mean shape. For all three procedures, within the respective Euclidean space, the corresponding tests then test for a common mean via the classical Hotelling $T^2$-test. Note that no two groups from the same tree are tested because of statistical dependence.

	       \begin{figure}[h!]
%	\begin{minipage}{0.6\textwidth}
	\centering
        	\includegraphics[width=0.7\textwidth, angle=-90]{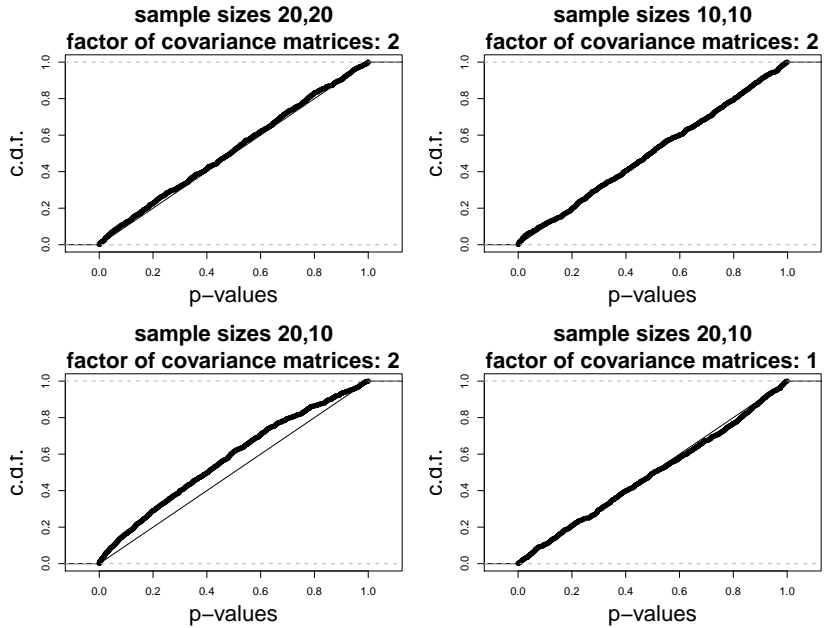}
% 	\end{minipage}
% 	\begin{minipage}{0.05\textwidth}
% 	~ 	
% 	\end{minipage}
% 	\begin{minipage}{0.3\textwidth}	
	\caption{\it Hotelling's $T^2$-test under nonnormality. The cumulative distributions of the empirical test statistics have been generated with $1,000$ repetitions of two groups of sizes as displayed in the respective headers, with three-dimensional deviates uniform in a 3D interval centered at the origin, whose covariance matrices (which are determined by the dimensions and rotation of the interval) are constant multiples of one another up to a random rotation (fixed for every display). The respective headers give the corresponding constant factors. \label{simulation-nonnormal:fig}}
%	\end{minipage}
       \end{figure}

	\paragraph{Robustness under nonnormality} Recall that the \emph{Hotelling} $T^2$-distribution is a generalization of a \emph{Student} $T$-distribution. As in the univariate case, the corresponding test statistic follows this distribution, if the coordinate data follow a multivariate normal distribution. Since the shape spaces considered are compact, obviously, we may never assume this central hypothesis. It is well known, however, that the corresponding statistic is robust to some extent under nonnormality, one condition being finite higher order moments. Clearly, this condition is met on a compact space. Even more, it is known that asymptotically the distribution of the corresponding statistic is unchanged under unequal change of covariances  if the ratio of sample sizes tends to 1, e.g. \citet[p. 462]{L97}. The simulation in Figure \ref{simulation-nonnormal:fig} illustrates robustness for fairly small sample sizes: under nonnormality and unequal sample sizes (bottom right display), and under nonnormality and unequal covariances with equal sample sizes (top row). The test is not robust under non-normality, however, if sample sizes and covariances differ considerably (bottom left display). 

	\paragraph{The test for common geodesics}
	Every first GPC computed above	deter\-mines a unique element in the space of geodesics $\Gamma(\Sigma_2^4)$ of $\Sigma_2^4$. Using the embedding of the space of pre-geodesics $O_2^H(2,3)$ into $\mathbb C^3\times \mathbb C^3$ as detailed in Section \ref{Asymp:scn}, these elements are orthogonally projected to the tangent space of the two group's Ziezold mean (an element of the space of geodesics $\Gamma(\Sigma_2^4)$) thus giving data in a Euclidean space. The corresponding null hypothesis is then
	\begin{center}{\it the temporal evolution of shape for every group follows a common geodesic.}\end{center}
	In other words, if $\gamma_i$ are the first GPCs %principal component population geodesics 
	of leaf growth of groups $i=1,2$ to be specified later, then the null hypothesis states $H_0: \gamma_1=\gamma_2$. This is a hypothesis on the mean geodesic of geodesics which can be tested by use of Theorem \ref{Zie_mean:thm}.

	\paragraph{Tests for common means} Following the classical scheme (e.g. \citet[Chapter 7]{DM98}), all shapes of the two groups considered are projected to the tangent space of their overall Procrustes mean giving Procrustes residuals with the null hypothesis,
	\begin{center}{\it the Procrustes tangent space coordinates of the temporal evolution\\ of shape for every leaf have the same Euclidean mean.}\end{center}
	\paragraph{A theoretical note on related tests}
	Taking instead the Procrustes residuals at the common Procrustes mean of the Procrustes means of every temporal evolution of shape would give a different test.	If all leaves considered have equally many individual shapes then this second test is very closely approximated by the first test. Otherwise, choosing suitable weights will give a close approximation. The goodness of the approximation can be numerically confirmed, it also follows from the fact that only one effect is tested, cf. \citet[pp.2-4]{HHM07}. Similar tests for static shape utilize intrinsic or Ziezold means, respectively. In fact, for hypotheses on non-degenerate three-dimensional shapes one could only test hypotheses building on intrinsic or Ziezold means (because Procrustes means may lie outside the manifold part such that the CLT is not applicable, cf. \cite{H_meansmeans_10}).

	\paragraph{Tests for common directions} Returning to the classical scheme, following \cite{MKMA00}, compute the Euclidean first principal component of each set of Procrustes residuals corresponding to the shapes of a single leaf's evolution pointing into the direction of growth. For the analysis of the unit length directions we proceed as described above. Within the Procrustes paradigm, the residual tangent space coordinates of these directions at their common residual mean closer to the data are projected  orthogonally to the Euclidean space of suitable dimension. The null hypothesis is then,
	\begin{center}{\it the Procrustes tangent space coordinates of the temporal evolution\\ of shape for every leaf share the same first Euclidean PC.} \end{center}
	This is the version of the \emph{parallel hypothesis} for this paper. If the static mean shapes of the two groups considered are different, then the common Procrustes mean depends on the ratio of the two sample sizes. Moreover, even for common static shape, the null hypothesis incorporates effects of curvature, cf. Figure \ref{two-arrows:fig}.

\section{Discriminating Canadian Black Poplars by Partial Observation of Leaf Growth}\label{App:scn}

	\begin{table}[b!]\centering
	{\footnotesize 
	\fbox{
	\begin{tabular}{cc|ccc} 
	dataset 1&dataset 2&geodesics&directions&means\\\hline
	  clone 1 young (21)& clone 2 old (11)&$0.75$&$1.6e-04$&$4.7e-08$\\
	  clone 1 young (21)& clone 2 young (11)&$0.97$&$0.82$&$0.71$\\
	  clone 2 young (11)&clone 1 old (20)&$0.066$&$0.0021$&$5.5e-07$\\
	  clone 1 old   (20)& clone 2 old (11)&$0.17$&$0.21$&$0.71$\\

\hline	\multicolumn{2}{c|}{correct classification of clones:}\\
% \hline	  \multicolumn{2}{c|}{at $95\,\%$-level}&$100.00\,\%$&$75.00\,\%$&$75.00\,\%$\\
% 	  \multicolumn{2}{c|}{at $99\,\%$-level}&$75.00\,\%$&$50.00\,\%$&$75.00\,\%$
%\hline	  
\multicolumn{2}{c|}{at $95\,\%$-level}&$100.00\,\%$&$50.00\,\%$&$50.00\,\%$\\
	  \multicolumn{2}{c|}{at $99\,\%$-level}&$100.00\,\%$&$50.00\,\%$&$50.00\,\%$\\ \hline
\hline
	  clone 1 young (21)& reference young (12)&$0.0012$&$0.65$&$6.5e-06$\\
	  clone 2 young (11)& reference young (12)&$0.0043$&$0.79$&$0.0015$\\
%	  clone 1+2 young (32)& reference young (12) &$0.00012$&$0.65 $ &$9.6-8 $\\

	  clone 1 young (21)& reference old (9)&$0.0067$&$0.0077$&$1.9e-07$\\
	  clone 2 young (11)& reference old (9)&$0.026$&$0.023$&$1.0e-05$\\
%	  clone 1+2 young (32)& reference old (9) &$0.015 $&$0.00095 $ &$7.1e-11 $\\

	  clone 1 old (21)& reference young (12)&$0.00022$&$0.0013$&$2.4e-06 $\\
	  clone 2 old (11)& reference young (12)&$0.0092$&$0.014$&$0.0046$\\
%	  clone 1+2 old (32)& reference young (12) &$ 1.6e-05$ &$ 8.3e-05 $ &$ 1.4e-08$\\

	  clone 1 old (21)& reference old (9)&$0.087 $&$0.023 $ &$0.0014$\\
	  clone 2 old (11)& reference old (9)&$0.021 $&$0.018 $ &$0.0068$\\
%	  clone 1+2 old (32)& reference old (9) &$0.045 $&$0.0024 $ &$3.0e-5 $\\

\hline	\multicolumn{2}{c|}{correct classification of the reference tree}\\
%\hline	
  \multicolumn{2}{c|}{at $95\,\%$-level}&$87.50\,\%$&$75.00\,\%$&$100.00\,\%$\\
	  \multicolumn{2}{c|}{at $99\,\%$-level}&$62.50\,\%$&$25.00\,\%$&$100.00\,\%$\\
% \hline	  \multicolumn{2}{c|}{at $95\,\%$-level}&$91.67\,\%$&$75.00\,\%$&$100.00\,\%$\\
% 	  \multicolumn{2}{c|}{at $99\,\%$-level}&$58.33\,\%$&$41.67\,\%$&$100.00\,\%$\\
\hline\hline	\multicolumn{2}{c|}{correct classification:}\\ \multicolumn{2}{c|}{ clones vs. reference tree}\\
\hline	  \multicolumn{2}{c|}{at $95\,\%$-level}&$93.75\,\%$&$62.50\,\%$&$75.00\,\%$\\
	  \multicolumn{2}{c|}{at $99\,\%$-level}&$81.25\,\%$&$37.50\,\%$&$75.00\,\%$
% \hline	  \multicolumn{2}{c|}{at $95\,\%$-level}&$95.83\,\%$&$62.50\,\%$&$75.00\,\%$\\
% 	  \multicolumn{2}{c|}{at $99\,\%$-level}&$79.17\,\%$&$45.83\,\%$&$75.00\,\%$
	 \end{tabular}}}
	\caption{\it Displaying $p$-values for several tests for the discrimination of clones from the reference tree via leaf growth (``young'' denotes the dataset comprising the first initial two  observations and ``old'' the dataset comprising the rest of the observations). For convenience, the sample size (number of different leaves followed over their growing period) of the corresponding data set is reported in parentheses.\label{initial-vs-rest:tab}}
	\end{table}

%	\paragraph{Discriminating and classifying growth using two initial observations.} 
%	Recall the data described in Section \ref{Forest:scn} and the tests introduced in Section \ref{Tests:scn}. 
	For the following tests, as the first groups called ``young'' in the following, the first two initial shapes from every leaf considered have been taken and the unique geodesic joining the two has been computed. For the second groups, called ``old'' in the following, for every leaf considered the first GPC of the rest of the shapes has been computed, if the  number of the rest of shapes exceeded 3. For this reason, the number of GPCs in some of these groups is possibly smaller then the number in corresponding groups of ``young''. To the respective groups the three tests introduced in Section \ref{Tests:scn} have been applied. The results are reported in Table \ref{initial-vs-rest:tab} in different order, however, since the tests for common geodesics and for common directions test for closely related concepts. %, the tests appear here in an order different from Section \ref{Tests:scn}. %; the ultimate test for common means may serve as a control.

	As visible in Table \ref{initial-vs-rest:tab}, the test for common geodesics (first data column) allows to discriminate  very well genetically different trees from genetically identical trees by the observation of growth over restricted time intervals only. Due to curvature (cf. Figure \ref{two-arrows:fig}) in the first box for genetically identical trees, whenever over these restricted time intervals the means differ (ultimate data column, i.e. when the growth of young leaves is compared to the growth of old leaves), then the directions also differ highly significantly (middle data column). For genetically different trees (third box) all of the group means differ highly significantly (ultimate data column) and most of the directions as well. Note that genetically different young leaves cannot be discriminated by their directions. This can be explained by comparison with the bottom right display of Figure \ref{rot-schwarz-ref-young-old-growth2008:fig}: leaf shapes of young leaves tend to be comparatively close to each other.

	\begin{figure}[h!]
	\centering
	 \includegraphics[angle=-90,width=0.5\textwidth]{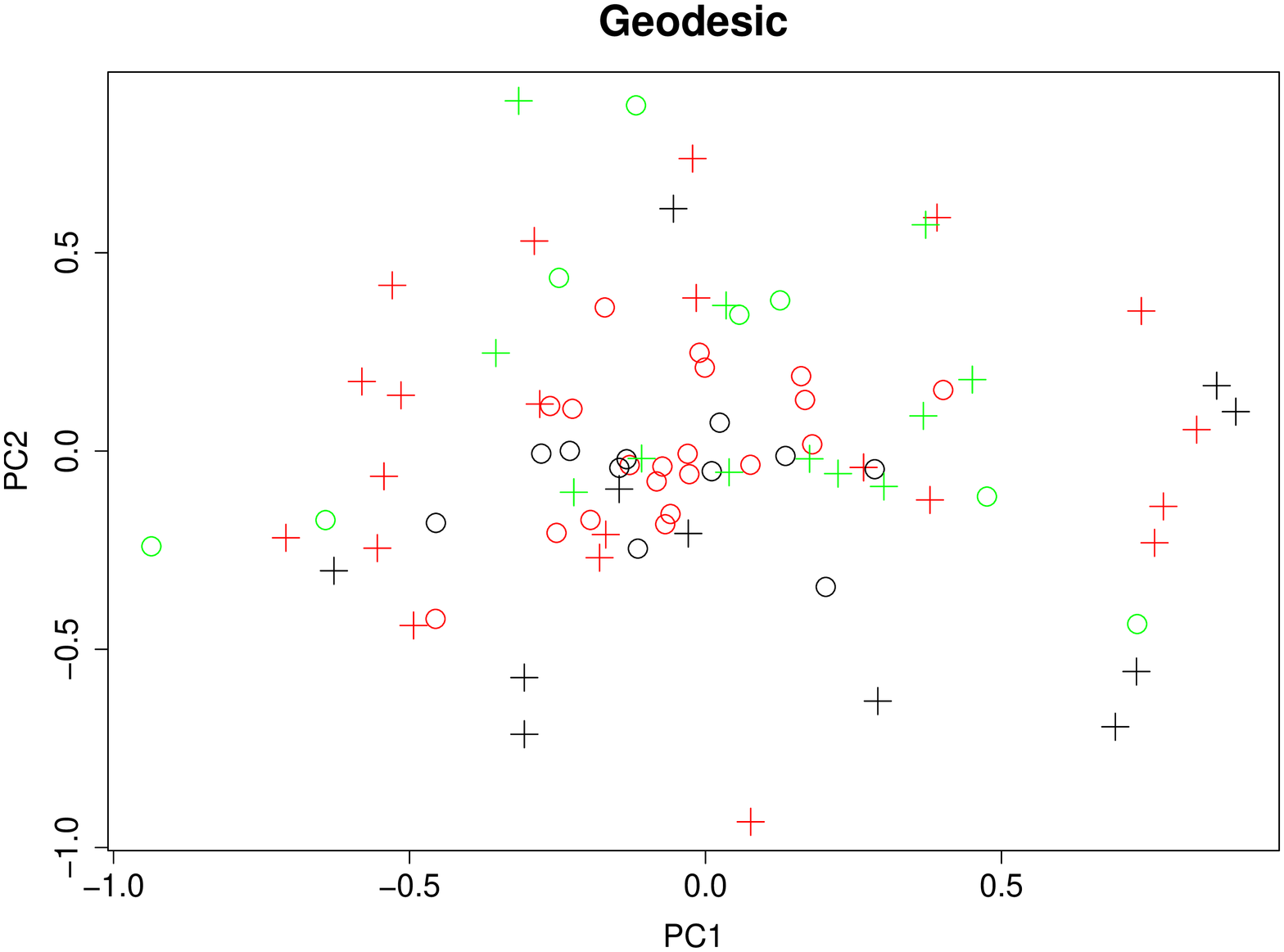}\\
	 \includegraphics[angle=-90,width=0.48\textwidth]{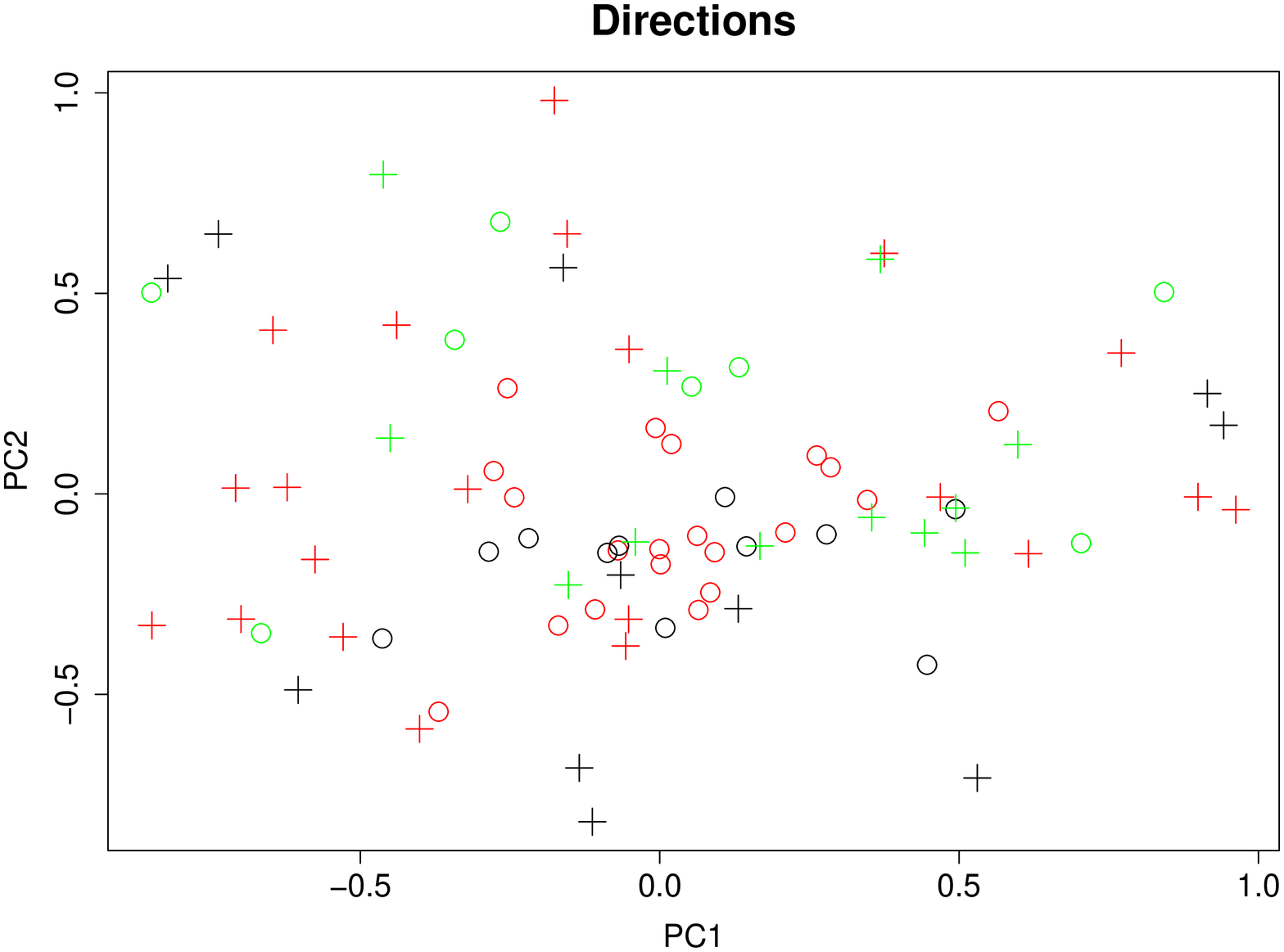}
	 \includegraphics[angle=-90,width=0.48\textwidth]{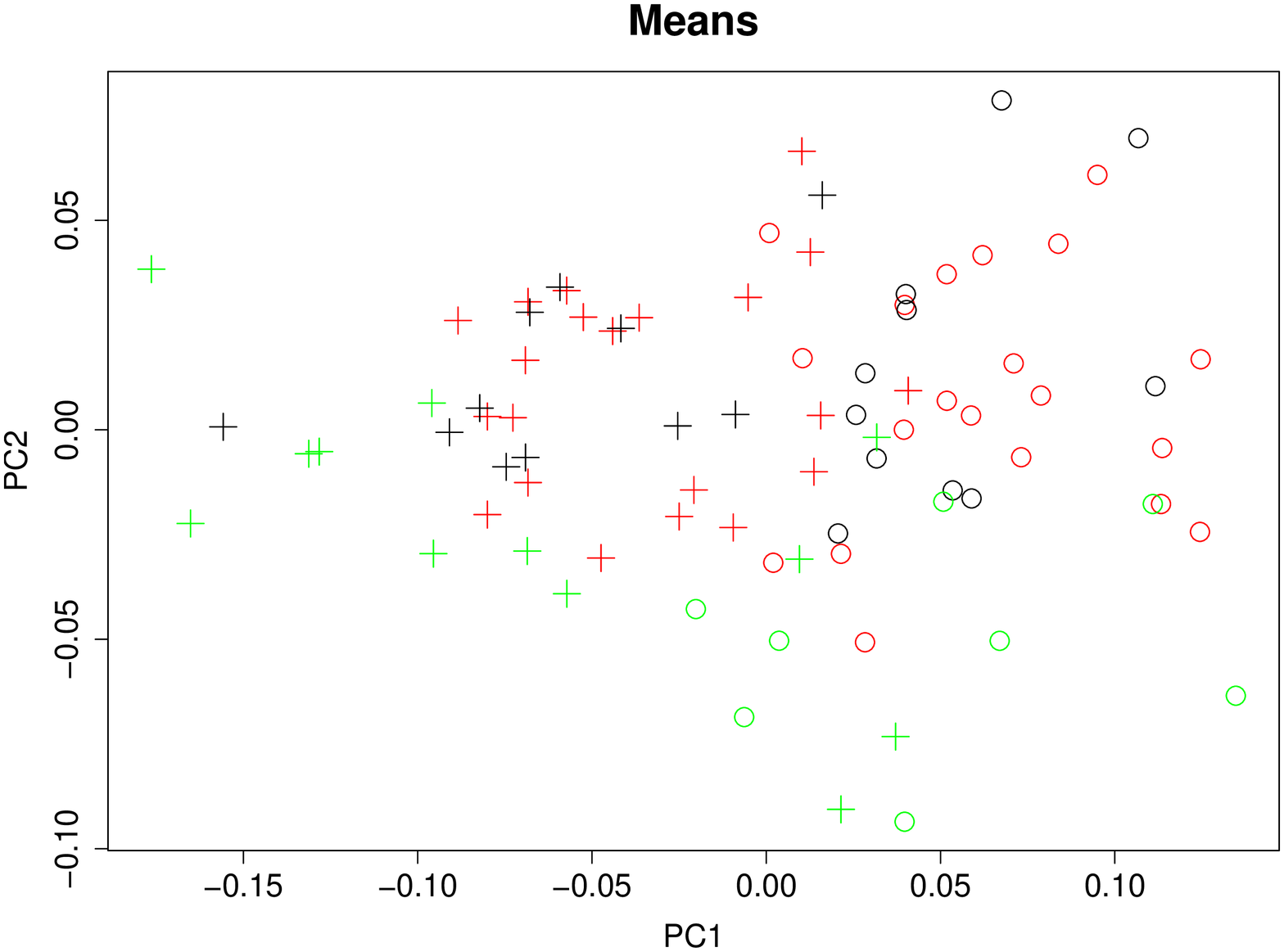}
	\caption{\it Projection of leaf shape growth from two clones (red for clone 1 and black for clone 2) and a reference tree (green) to the dominating coordinates of the corresponding tangent space at the overall mean as detailed in Section \ref{Tests:scn}. Each cross represents a single leaf's initial shape evolution over two observations (young), each circle represents the rest of the leaf's shape evolution during its observed growing period (old). Top: % left: 
	GPCs projected to the tangent space at the Ziezold mean geodesic. %Top right: parallel transplants of geodesic unit directions at the points nearest to the overall intrinsic mean. 
	Bottom row: all shapes have been projected to the tangent space of the overall Procrustes mean. Bottom left: unit directions of first Euclidean PCs. Bottom right: Euclidean means.
	\label{rot-schwarz-ref-young-old-growth2008:fig}}
	\end{figure}
	 
	Figure \ref{rot-schwarz-ref-young-old-growth2008:fig} depicts the first two dominating coordinates (explaining between $80\,\%$ and $90\,\%$ of the total variation) of a tangent space projection of the overall dataset (young and old) for clones and the reference tree. In contrast to the test for common geodesics (cf. top display), almost all of the different groups (young vs. old of clone 1, clone 2, and the reference tree) can be discriminated by the directional data (bottom left display, e.g. the red and black crosses tend to lie on the l.h.s., the red and black circles on the r.h.s.). % while almost none of the groups can by discriminated by the ``parallelity'' data (top right display). 
	In all of the first two displays (top and bottom left), for the clones, the variance of the ``young'' data appears slightly larger than the variance of the ``old'' data. For the reference tree, however, this is not  the case. This effect is not visible when observing means only (bottom right display). As discussed in Section \ref{Tests:scn}, unequal covariances may be troublesome w.r.t. to the validity of the Hotelling $T^2$-test employed if sample sizes are not approximately equal. Considering in Table \ref{initial-vs-rest:tab} only the samples of similar sizes $9,11$ and $12$, however, comparable classification results are obtained.

	\paragraph{Conclusion} From this study we conclude the following.
	\begin{enumerate}
	 \item[(a)] Clone and reference tree can be discriminated by partial observations of leaf shape growth not necessarily covering the same interval of the growing period via the test for common geodesics. This is not possible via a test for common means (due to temporal change of shape) or common directions (due to curvature).
	\item[(b)] The ``geodesic hypothesis'' has been validated by use of	 the test for common geodesics, notably it could not have been validated using the test for common directions.
	\item[(c)] For a statistical prediction of future leaf shape growth, two initial observations suffice.
	\end{enumerate}

	\paragraph{Application} Let us elaborate on one consequence of conclusion (a). Suppose that we have several leaf shape growth data of clones and the reference tree of short but arbitrary time intervals (young, old, intermediate, etc.). Then most likely the test for common means will not be able to identify the clone from the reference tree, because the mean shapes of the different time intervals will most likely  be different even for the same leaf considered. Similarly, due to curvature the test for common directions will fail, unless all data are jointly highly concentrated. It is only the test for common geodesics that may furnish the desired discrimination. %Moreover, since only two initial observations suffice for discrimination, it seems that elaborate models incorporating more effects cannot be statistically justified for the discrimination problem at hand.

\section{Discussion}\label{Disc:scn}

	% the first test for a geodesic hypothesis
	The \emph{geodesic hypothesis} of \cite{LK00} -- its scope and limitations have been discussed in Section \ref{Forest:scn} -- %(cf. Section \ref{Forest:scn}) 
	has been corroborated in many scenarios of biological growth. In this paper a  \emph{test for common geodesics} for this hypothesis has been devised and successfully applied to the problem of discriminating poplar leaf growth based on two initial observations only. For computational feasibility, not the concept of an intrinsic mean geodesic but rather that of a Ziezold mean geodesic has been employed. One can thus call the test devised a \emph{semi-intrinsic test}. The semi-intrinsic test for comman geodesics has been compared to a test for common directions building on directions in the space of Procrustes residuals (following \cite{MKMA00}). This is essentially a \emph{non-intrinsic test} because it linearizes the shape space and not the space of shape descriptors tested for. It turned out that for the discrimination task at hand, curvature present rendered this test ineffective. % has failed for the classification task. 
	The author is not aware of any other test for the geodesic hypothesis in the literature.

	 %In order to obtain an intrinsic version,  namely a test for \emph{parallelity}, different concepts of parallelity have been introduced which are, however, local in nature. This test added the insight that while shorter geodesic segments (either initial or the rest of the growth) observed may be parallel, the entire segments (initial and the rest of growth) are not parallel in the sense defined by the test. Using the other definition of parallelity discussed above gives comparable results. Since there is no concept of parallelity of geodesics on non-flat spaces, there can be no intrinsic test for a parallel hypothesis (Section \ref{Forest:scn}), in fact one can say that there is no intrinsic parallel hypothesis. Two doable workarounds have been proposed but, in view of model building, they are not quite satisfactory.
	
%	Problem 2 of Section \ref{Forest:scn} has found the answers: for short periods of growth there may be some parallelity, for the entire growing period there is hardly any. More research is necessary, e.g. to determine whether the negative answer is due to the fact that there is globally no parallelity 

%	For data sufficiently concentrated, note that the geodesic hypothesis implies the parallel hypothesis. 
	
	In this work we have considered two types of mean first GPCs, one defined by a sample of GPCs with underlying samples of random shapes, which -- like growth patterns -- are obviously dependent. For independent sampling the other mean first GPC has been defined directly by the shape data. Since $\rho$ (defining the latter) is different from $\delta$ (defining the former, cf. Section \ref{SS:scn}) -- as a manifestation of `inconsistency' (cf. \cite{KM97}; \cite{H_Procrustes_10}) -- the limit of mean random geodesic of geodesics and the population geodesic of shapes may be different as well. Studying their relationship, however, may provide further insight. %Challenging in this context appears a 

% 	In this work we have considered two types of random first GPCs, one as a defined by a sample of random shapes, which -- like growth patterns -- is obviously dependent. For independent sampling in other scenarios, the corresponding random geodesics %, in consequence of %\citet[Theorem 3.4 and 3.5]{H_Procrustes_10} 
% 	converge a.s. to the corresponding population first GPC, if unique (cf. Section \ref{SS:scn}). For the Central Limit Theorem (CLT) %(cf. \citet[Theorem 3.8]{H_Procrustes_10}), 
% 	to hold smoothness of $\rho^2$ from Section \ref{SS:scn} needs to be established. Since $\rho$ is different from $\delta$ (also in Section \ref{SS:scn}) -- as a manifestation of `inconsistency' (cf. \cite{KM97}; \cite{H_Procrustes_10}) -- the limit of mean random geodesic and the population geodesic may be different as well. Studying their relationship, however, may provide further insight. %Challenging in this context appears a generalization of the result of \cite{A63} giving explicitly the covariances for the asymptotic distribution of Euclidean PCs.

%	In a previous research (\cite{HHM09}) developing an intrinsic version of MANOVA, we have relied on the assumption of equal covariances under parallel transport,  also assumed by \cite{FS08} in a similar context. This assumption is certainly debatable. For this research, however, this assumption is no longer necessary (cf. Section \ref{Tests:scn}).

	In conclusion let us ponder on extensions and generalizations of this research. One may view all shape descriptors as generalized Fr\'echet means on suitable spaces. For geodesics on Kendall's planar shape spaces, we have provided an explicit framework using a Ziezold mean geodesics which can be computed fairly easy. Straightforward but considerably more complicated is the use of intrinsic mean geodesics. At this point we note that the space of generalized geodesics on Kendall' shape space for dimension $m\geq 3$ ceases to be a manifold. Like a shape space, it can be viewed as the quotient of a Riemannian manifold modulo a Lie group action. In contrast to Kendall's shape spaces, the top space itself (a submanifold of a Grassmannian) admits two canonical geometries (cf. \cite{EAS98}). For the embedding underlying the Ziezold mean geodesic, we have used the simpler of the two. Possibly, this framework extends to other one-dimensional shape descriptors, such as arbitrary circles on spheres (cf. \cite{JFM_sinica_2010}), or more generally, the family of constant curvature curves, as well as to higher dimensional shape descriptors (for geodesic descriptors cf. \cite{HHM07}, for non-geodesic descriptors cf.  \cite{JFM2011}). 
	Thus, (semi)-intrinsic inference on any of such descriptors may be possible. 

	One final word of caution: even for fairly simple spaces such as the torus or the surface of an infinite cylinder, the canonical topology of the space of geodesics is non-Hausdorff (cf. \cite{BeemParker91}) and thus may not admit any meaningful statistical descriptors.

\section*{Acknowledgments}
	The author would like to thank his colleagues from the Institute for Forest Biometry and Informatics at the University of G\"ottingen whose interest in leaf growth modeling % clone discrimination by leaf shapes 
	prompted this research. In particular he would like to thank Michael Henke for supplying with the leaf data. Moreover, he is indebted to Thomas Hotz for discussing statistical issues and to David Glickenstein for a comment on geometric aspects. Also, %Moreover, 
	the author is indebted to the anonymous referees %and two associate editors 
	for their very valuable comments, some of which have been literally adopted.
	 %Also, the author gratefully acknowledges support by DFG grant MU 1230/10-1 and GRK 1023. 

%\section*{Appendix}
\appendix

\section{Strong Consistency}\label{consist:app}	
	
	For this section suppose that $X, X_1,X_2,\ldots$ are i.i.d. random elements mapping from an abstract probability space $(\Omega,\cal A,{\cal P})$ to a topological space $Q$ equipped with its Borel $\sigma$-field; $(P,d)$ denotes a topological space with distance $d$. 

	\begin{Def}\label{Frechet_means:def} For a continuous function $\rho:Q\times P \to [0,\infty)$  define the \emph{set of  population Fr\'echet $\rho$-means of $X$ in $P$} by
	$$ E^{(\rho)}(X) = \argmin_{\mu\in P} \mathbb E\big(\rho(X,\mu)^2\big) %\int_Q \rho^2(\mu,q)\,dP_X(q)
	\,.$$
	For $\omega\in \Omega$ denote by
	$$ E^{(\rho)}_n(\omega) = \argmin_{\mu\in P} \sum_{j=1}^n \rho\big(X_j(\omega),\mu\big)^2$$
	the \emph{set of sample Fr\'echet $\rho$-means}.
	\end{Def}

	By continuity of $\rho$, the mean sets are closed random sets. For our purpose here, we rely on the definition of \emph{random closed sets} as introduced and studied by \cite{Choq54}, \cite{Kend74} and \cite{Math75}. Since their original definition for $P=Q, \rho =d$ a metric by \cite{F48} such means have found much interest.

	We will work with the following two definitions of \emph{strong consistency}, each has been coined as such for metrical Fr\'echet means by the respective authors.

	\begin{Def}\label{SLLN:def} Let $E^{(\rho)}_n(\omega)$ be a random closed set and $E^{(\rho)}$ a deterministic closed set in a space with distance, $(P,d)$. We then say that 
	\begin{enumerate}
	 \item[(ZC)] $E^{(\rho)}_n(\omega)$ is a \emph{strongly consistent estimator in the sense of \citet{Z77}} of $E^{(\rho)}$ if  for almost all $\omega\in \Omega$
% 	$\cap_{n=1}^\infty A_n(\omega) \subset E$ for all $\omega\in \Omega$ a.s.
	$$\bigcap_{n=1}^\infty \overline{\bigcup_{k=n}^\infty E^{(\rho)}_k(\omega)} \subset E^{(\rho)}$$
	 \item[(BPC)] $E^{(\rho)}_n(\omega)$ is a \emph{strongly consistent estimator in the sense of  \citet{BP03}} of $E^{(\rho)}$ if  $E^{(\rho)}\neq \emptyset$ and if for every $\epsilon >0$ and almost all $\omega\in \Omega$ there is a number $n=n(\epsilon,\omega)>0$ such that
%	$A_n(\omega) \subset \{q\in Q: d_Q(E,q)\leq \epsilon\}$
	$$\bigcup_{k=n}^\infty E^{(\rho)}_k(\omega) \subset \{p\in P: d(E^{(\rho)},p)\leq \epsilon\}\,.$$
	\end{enumerate}
	\end{Def}
 
 	For quasi-metrical means on \emph{separable} (i.e. containing a dense countable subset) quasi-metrical %(with a non-negative symmetric distance satisfying the triangle inequality) 
 	spaces, \cite{Z77} proved (ZC).  
 	For metrical means on spaces that enjoy the stronger \emph{Heine-Borel property} (i.e. that every bounded closed set is compact), \cite{BP03} proved (BPC) . 
	%Note that both properties (ZC) and (BPC) have been called ``strong consistency'' by their respective authors. 
	One may argue from a statistical point of view that for ``consistency'' one would want to have equality of points, and if not possible, at least an equality of sets (the set of 3D Procrustes means always contains at least two antipodal points, e.g. \cite{H_Procrustes_10}), rather than an inclusion only. Even though it would be interesting to construct an example with strict inclusion, it seems that this case has no relevance in applications.

	In order to generalize Ziezold's and the Bhattacharya-Patrangenaru Strong Consistency Theorem we introduce two properties. A \emph{continuity} property in the second argument \emph{uniform} over the first argument -- a consequence of the triangle inequality if $\rho$ is a quasi-metric -- and a version of \emph{coercivity} in the second argument -- again valid if $\rho$ is a quasi-metric:
	\begin{eqnarray}
	 \label{unif_cont}
	&&\left.\begin{array}{l}\mbox{for every $x\in Q,p\in P$ and $\epsilon>0$ there is a $\delta=\delta(\epsilon,p)>0$ }\\ 
	\mbox{such that $|\rho(x,p') -\rho(x,p)| <\epsilon$ for all $p'\in P$ with $d(p,p')<\delta$}\\ 
	 \end{array}\right\} \\
	\label{mild_coerc}
	&&\left.\begin{array}{l}\mbox{there are $p_0\in P$ and $C>0$ such that ${\cal P}\{\rho(X,p_0)<C\}>0$ and}\\ 
	\mbox{that such that  for every sequence $p_n\in P$ with $d(p_0,p_n)\to \infty$}\\ \mbox{there is a sequence $M_n \to \infty$ with $\rho(x,p_n) > M_n$ for all $x\in Q$ }\\ \mbox{with $\rho(x,p_0)<C$. Moreover, if $p_n\in P$ with $d(p*,p_n)\to \infty$} \\\mbox{for some $p*\in M$, then $d(p_0,p_n)\to \infty$.}\end{array}\right\} 
	\end{eqnarray}

	\begin{Th}[Ziezold's Strong Consistency]\label{ext_Ziezold:thm}
	 Let $\rho : Q\times P \to [0,\infty)$ be a continuous function on the product of a topological space with a separable space with distance $(P,d)$. Then strong consistency holds in the Ziezold sense (ZC)  for the set of Fr\'echet $\rho$-means on $P$ if 
	\begin{enumerate}
	 \item[(i)] $X$ has compact support, or if
	 \item[(ii)] $\mathbb E\big(\rho(X,p)^2\big)<\infty $ for all $p\in P$ and $\rho$ is uniformly continuous in the second argument in the sense of (\ref{unif_cont}).
	\end{enumerate}
	\end{Th}

	\begin{proof}
	Obviously under $(i)$ for every $p\in P$ and $\epsilon  >0$ we may assume that there is $\delta=\delta(\epsilon,p)$ such that $|\rho(X,p') -\rho(X,p)|<\epsilon$ a.s. if $d(p,p')<\delta$. With this is mind, it suffices to prove the assertion under (ii).
	
	 For $\omega \in \Omega$, $p\in P$ set
	\begin{eqnarray}\label{ziez:hyp} 
	&\left.\begin{array}{rclcrcl}F_n(p) &=& \frac{1}{n}\sum_{j=1}^n \rho(X_j(\omega),p)^2,&&F(q) &=& \mathbb E\big(\rho(X,p)^2\big),\\ \lw_n &=& \inf_{p\in P}F_n(q),&&\lw &=& \inf_{p\in P}F(p),\\
	E_n &=&\{p\in P: F_n(p) = \lw\},&& E&=&\{p\in P: F(p)=\lw\}\,.\end{array}\right\}\end{eqnarray}
	We now follow the steps laid out in \cite{Z77}. Let $p_1,p_2,\ldots$ be dense in $P$. 
	From the usual Strong Law of Large Numbers in $\mathbb R$ we have sets $A_k\subset \Omega$, ${\cal P}(A_k) =1$ such that
	$F_n(p_k) \stackrel{n\to \infty}{\to} F(p_k)$
	for every $k=1,2,\ldots$ and $\omega \in A_k$. Setting $A:= \cap_{k=1}^\infty A_k$ we have hence for all $p=p_k, k=1,2,\ldots$
	\begin{eqnarray}\label{Ziez1:eq} 
	F_n(p) \stackrel{n\to \infty}{\to} F(p) \mbox{ for all }\omega \in A,~{\cal P}(A)=1\,.
	\end{eqnarray}
	Next, let $p,p' \in P$. Setting $f(q,p',p) :=\rho(q,p') - \rho(q,p)$ we have then
	\begin{eqnarray}\nonumber \label{Ziezolds_uni:est}
	 \big|F_n(p') - F_n(p)\big| &\leq& \frac{1}{n}\sum_{j=1}^n\big(\rho(X_j,p') + \rho(X_j,p)\big)\big|\rho(X_j,p') - \rho(X_j,p)\big|\\
	&=& \left\{\begin{array}{l}
	\frac{1}{n}\sum_{j=1}^n\big(2\rho(X_j,p') + f(X_j,p,p') \big) |f(X_j,p,p')| \\[0.2cm]
	\frac{1}{n}\sum_{j=1}^n\big(2\rho(X_j,p) + f(X_j,p',p) \big) |f(X_j,p,p')| 
	\end{array}\right.
	\end{eqnarray}
	W.l.o.g. we may suppose that $p_k \to p \in P$. In consequence from the top line of (\ref{Ziezolds_uni:est}) for $p'=p_k$, :
	\begin{eqnarray*}
	& \frac{1}{n}\sum_{j=1}^n \rho(X_j(\omega),p_k)^2 - \frac{1}{n}\sum_{j=1}^n\big(2\rho(X_j,p_k)  + |f(X_j,p,p_k)| \big) |f(X_j,p,p_k)|  
	\\
	&\leq \frac{1}{n}\sum_{j=1}^n\rho(X_j,p)^2\big) \leq\\
	& \frac{1}{n}\sum_{j=1}^n \rho(X_j(\omega),p_k)^2 + \frac{1}{n}\sum_{j=1}^n\big(2\rho(X_j,p_k)  + |f(X_j,p,p_k)| \big) |f(Y_j,p,p_k)|  
	\end{eqnarray*}
	Taking the expected value, i.e. employing (\ref{Ziez1:eq}) at $p_k$, by hypothesis (i) or (ii) as explained in the beginning of the proof, for arbitrary $\epsilon >0$ we may assume that there is a $\delta>0$ such that for all $d(p_k,p)<\delta$
	\begin{eqnarray*}
	& \mathbb E\big(\rho(X,p_k)^2\big) - \Big(2\mathbb E\big(\rho(X,p_k)\big) + \epsilon \Big) \epsilon 
	\\
	&\leq \liminf_{n\to \infty}\frac{1}{n}\sum_{j=1}^n\rho(X_j,p)^2 \leq \limsup_{n\to \infty}\frac{1}{n}\sum_{j=1}^n\rho(X_j,p)^2\leq\\
	& \mathbb E\big(\rho(X,p_k)^2\big) + \Big(2\mathbb E\big(\rho(X,p_k)\big) + \epsilon \Big) \epsilon \,.
	\end{eqnarray*}
	Letting $\epsilon \to 0$ we can choose a subsequence $p_{k_i} \to p$, hence, this yields the validity of (\ref{Ziez1:eq}) for all $p\in P$.

	Let us now extend (\ref{Ziez1:eq}) to 
	\begin{eqnarray}\label{Ziez2:eq}&& 
	F_n(p_n) \stackrel{n\to \infty}{\to} F(p) \mbox{ for all sequences $p_n\to p$ and } \omega \in A,~{\cal P}(A)=1\,.
	\end{eqnarray}
	Utilizing the bottom line from (\ref{Ziezolds_uni:est}) yields 
	$$\big|F_n(p_n) - F_n(p)\big| \leq\frac{1}{n}\sum_{j=1}^n\big(2\rho(X_j,p) + |f(X_j,p_n,p)|  \big) |f(X_j,p,p_n)|  \to 0\,.$$
	%as above.
	Hence as desired $\big|F_n(p_n) - F(p)\big| \leq \big|F_n(p_n) - F_n(p)\big| + \big|F(p) - F_n(p)\big| \to 0$. % by continuity.
	
	Finally let us show
	\begin{eqnarray}\label{Ziez3:eq} \mbox{if $\cap_{n=1}^\infty \overline{\cup_{k=n}^\infty E_n}\neq \emptyset$ then $\lw_n \to \lw$}\end{eqnarray}
	Note that, the assertion of the Theorem is trivial in case of $\cap_{n=1}^\infty \overline{\cup_{k=n}^\infty E_n}=\emptyset$. Otherwise, for ease of notation let $B_n := \cup_{k=n}^\infty E_k$, $\overline{B_n}\searrow B:=\cap_{n=1}^\infty \overline{B_n}$, $b \in B$. Then $b \in \overline{B_n}$ for all $n\in\mathbb N$. Hence, there is a sequence $b_n \in B_n$, $b_n \to b$. Moreover there is a sequence $k_n$ such that
	 $b_n =p_{k_n} \in E_{k_n}$ for a suitable $k_n \geq n$. 
	Then $\lw_{n_k}=F_{n_k}(p_{n_k})\to F(b)\geq \lw$ by (\ref{Ziez2:eq}). On the other hand by (\ref{Ziez1:eq}) for arbitrary fixed $p\in P$ there is a sequence $\epsilon_{n}\to 0$ such that
	$F(p)\geq F_n(p)- \epsilon_{n}\geq \lw_n - \epsilon_{n}$. First letting $n\to \infty$ and then considering the infimum over $p\in P$ yields 
	$$\lw \geq\limsup_{n\to\infty} \lw_n\,.$$
	In consequence $\lw_n \to \lw =F(b)$. In particular we have shown that $b\in E\neq \emptyset$ thus completing the proof.
	\end{proof}

	\begin{Th}[Bhattacharya-Patrangenaru's Strong Consistency]\label{ext_BP:thm}	
	Suppose that (ZC) (Ziezold's strong consistency) holds for a continuous function $\rho: Q\times P\to \mathbb [0,\infty)$ on the product of a topological space with a %complete 
	space with distance $(P,d)$. If additionally $\emptyset \neq E^{(\rho)}$, %$\mathbb E\big(d(X,q)^2\big)<\infty$ for some $q\in Q$, 
	$\overline{\cup_{n=1}^\infty E^{(\rho)}_n(\omega)}$ enjoys the Heine-Borel property for almost all $\omega\in \Omega$ and the coercivity condition (\ref{mild_coerc}) in the second argument is satisfied then 
	the property of strong consistency (BPC) in the sense of Bhattacharya and Patrangenaru  is valid.
	\end{Th}

	\begin{proof} 
	We use the notation of the previous proof. %, in particular (\ref{ziez:hyp}). 
	%In order to see uniform strong consistency, 
	Consider a sequence $q_n \in E_n$ determined by
	$$d(p,E) = \max_{p\in E_n}d(p,E)=:r_n\,.$$
	If $r_n \not\to 0$ we can find a subsequence $n_k$ such that $r_{n_k}\geq r_0 >0$. In consequence, every accumulation point of $p_{n_k}$ has positive distance to $E$, a contradiction to strong consistency. Hence either $r_{n}\to 0$ or there are no accumulation points. 

	We shall now rule out the case that there are no accumulation points. In view of the Heine-Borel property, this case can only occur for $r_n \to \infty$. %Then we adapt the argument of \citet[pp. 9--10]{BP03}. 
	Under condition (\ref{mild_coerc}) there is $p_0\in P$, a subsequence $k(n)$ and for given $n$ a subsequence $n_1,\ldots,n_{k(n)}$ of $1,2,\ldots,n$ such that $\rho(X_{n_j},p_0) < C$ for all $j=1,\ldots, k(n)$, a.s., and 
	$$\frac{k(n)}{n} \to {\cal P}\{\rho(X,p_0)<C\} >0\,.$$
	Also, by hypothesis and condition (\ref{mild_coerc}), $d(p_0,p_n)\to \infty$.
	%Using the notation of (\ref{ziez:hyp}) 
	Moreover by condition (\ref{mild_coerc}), there is a sequence $M_n \to \infty$ with
	\begin{eqnarray*}
	 \lw_n ~=~F_n(p_n) &\geq & \frac{1}{n}\sum_{j=1}^{k(n)} \rho( X_{n_j}(\omega),p_n)^2~>~\frac{k(n)}{n}\,M^2_n \to \infty~~a.s.
	\end{eqnarray*}
	On the other hand for any fixed $p\in E$, we have the Strong Law on $\mathbb R$,
	$\lw_n \leq F_n(p) \to F(p) = \lw \mbox{ a.s.}\,,$
	yielding a contradiction.

	\end{proof}
\section{Proof of Theorem 2.2 %\ref{rhosq-smoot:th}
	}\label{rhosq-smoot:ap}
	\begin{proof}
	Consider for fixed $p\in S_2^k$ the non-negative $O(2)$-invariant smooth function on $O_2^H(2,k-1)$ defined by 
	$$f_p(x,v) := \Big(\arccos\sqrt{\langle p,x\rangle^2 + \langle p,v\rangle^2}\Big)^2\,.$$
	Then the condition
	$$ f_p(\xi x,\xi v) = \min_{e^{it}\in S^1} f_p(e^{it}x,e^{it}v) = \arccos\left(\max_{e^{it}\in S^1} \left(\langle e^{it}p,x\rangle^2 + \langle e^{it}p,v\rangle^2\right)\right) $$
	defines two smooth explicit branches $\pm \xi : O_2^H(2,k-1) \setminus M^0 \to S^1$ (cf. \citet[Theorem 4.3]{HT06}), outside the singularity set 
	$$M^0 = \{(x,v) \in O_2^H(2,k-1): D(x,v)=0 =A(x,v)^2- B(x,v)^2 \}$$
	using the notation from \cite{HT06}:
	\begin{eqnarray*}
	 A(x,v)^2 &=& \langle x,p\rangle^2 + \langle v,p\rangle^2\,,\\ 
	 B(x,v)^2 &=& \langle x,ip\rangle^2 + \langle v,ip\rangle^2\,,\\
	 D(x,v) &=& 2\Big(\langle x,p\rangle \langle x,ip\rangle +\langle v,p\rangle \langle v,ip\rangle\Big)\,.
	\end{eqnarray*}
	One verifies that $M^0$ is bi-invariant, i.e. invariant under the left action of $S^1$ and the right action of $O(2)$. Hence,
	on $O(2)\backslash \Big(O_2^H(2,k-1)\setminus M^0\Big)/S^1$, the square of the distance $\rho([p],[\gamma])$ agrees with the value of the bi-invariant function $(x,v) \mapsto f_p(\xi(x,v)x,\xi(x,v)v)$, hence 
	$$\frak{p}\circ\gamma\mapsto \rho^2([p],\frak{p}\circ\gamma)$$
	is smooth on $O(2)\backslash \Big(O_2^H(2,k-1)\setminus M^0\Big)/S^1$. 

	Finally, we show that the geodesics in $\Gamma(\Sigma_2^k)$ determined by $M^0$ are at least  $\pi/4$ away from $[p]$. Obviously, geodesics determined by $A(x,v)^2=0 = B(x,v)^2$ have distance $\pi/2$ to $[p]$. Suppose now that $(x,v)\in M^0$ with $A(x,v)^2>0$. W.l.o.g. assume that $\langle x,p\rangle \neq 0$. Then
	$$ B(x,v)^2 = \frac{\langle v,ip\rangle^2}{\langle x,p\rangle^2}\, A(x,v)^2$$
	which implies that $\langle v,ip\rangle^2=\langle x,p\rangle^2$. In consequence we have also $\langle v,p\rangle^2=\langle x,ip\rangle^2$ and, in particular, ${\rm sign}\big(\langle x,p\rangle \langle v,ip\rangle\big) =-{\rm sign}\big(\langle x,ip\rangle \langle v,p\rangle\big)=:\epsilon$. 
	Then, we have for the shape distance $\rho$ to $[p]$ of shapes along the geodesic $\gamma$ through $[x]$ with initial velocity $v$ at $[x]$ that
	\begin{eqnarray*}
	\cos\rho([p],\gamma(s)) &=& \max_{0\leq t< 2\pi}\langle p\cos t + ip \sin t,x\cos s+v\sin s\rangle\\
	&=&  \max_{0\leq t< 2\pi}\Big(\langle x,p\rangle \cos (t-\epsilon s) + \langle v,p\rangle \sin (t-\epsilon s)\Big)\\
	&=&\max\Big(|\langle x,p\rangle|,|\langle v,p\rangle| \Big)%~\geq~\frac{1}{\sqrt{2}} 
	\end{eqnarray*}
	is constant, giving as desired $\rho([p],\gamma)\geq \pi/4$. % since the same is true for the orthogonal projection of $[p]$ to $\Sigma_2^3$ spanned by $\gamma$ which is a two sphere with radius $1/2$..
% 
% 	to $[p]$
% 	of shapes along the geodesic $\gamma$ through $[x]$ with initial velocity $v$ at $[x]$ is constant. Suppose now that $[\widetilde{p}]$ is the orthogonal projection of $[p]$ to the complex projective subspace of $\Sigma_2^k$ through $[x]$, spanned by $v$ and $iv$ which is a a two-sphere of radius $1/2$, hence  $\rho([p],\gamma)\geq \rho([\widetilde{p}],\gamma) = \pi/4$.
% 
% 	 
% 	$\arccos \sqrt{A(x,v)^2}$to $[p]$  $x\cos t + v\sin t$ has distance  to $p$ constant in $t$. On the 2-sphere spanned by $x,v$ and $p$ this is only possible if $\langle x,p\rangle = 0 = \langle v,p\rangle$, a contradiction to $\langle x,p\rangle \neq 0$. Hence $A(x,v)^2=0$, completing the proof. 
	\end{proof}

\bibliographystyle{elsart-harv}
\bibliography{shape,botany}
\end{document}